\documentclass[a4paper,12pt]{article}
\usepackage{amsmath,amsfonts,amssymb,amsthm,amscd}
\usepackage{enumerate}
\usepackage{bm}

\usepackage{xspace}
\usepackage{graphicx}
\usepackage{xfrac}

\newcommand{\measurerestr}{%
\mathop{\lfloor}%
}

\newcommand{\R}{\mathbb{R}}

\newcommand{\Z}{\mathbb{Z}}
\newcommand{\W}{\mathcal{W}}
\renewcommand{\wp}{\widetilde{p}}

\newcommand{\rd}{\mathrm{d}}
\newcommand{\e}{\mathrm{e}}

\newcommand{\eps}{\epsilon}
\newcommand{\bydef}{\mathrel{\mathop{:=}}}

\newcommand{\ini}{\mathrm{I}}
\newcommand{\weakstar}{\overset{*}{\rightharpoonup}}
\newcommand{\A}{\ensuremath{\mathbb{A}}\xspace}
\newcommand{\B}{\ensuremath{\mathbb{B}}\xspace}

\renewcommand{\H}{\mathcal{H}}
\newcommand{\diag}[1]{\ensuremath{\mathop{\mathrm{Diag}}[#1]}}

\newcommand{\bx}{\mathbf{x}}

\newcommand{\bw}{\mathbf{w}}
\newcommand{\bz}{\mathbf{z}}

\newcommand{\bv}{\mathbf{v}}
\newcommand{\be}{\mathbf{e}}

\newcommand{\bp}{\mathbf{p}}
\newcommand{\bL}{\mathbf{L}}
\newcommand{\bq}{\mathbf{q}}
\newcommand{\bR}{\mathbf{R}}
\newcommand{\bpi}{\boldsymbol{\pi}}
\newcommand{\bpsi}{\boldsymbol{\psi}}

\newcommand{\bzero}{\mathbf{0}}

\newcommand{\tsp}{\mathsf{s}}

\newcommand{\degdiff}{\Theta}

\newcommand{\calA}{\mathcal{A}}

\newcommand{\bone}{\mathbf{1}}

\newcommand{\bF}{\mathbf{F}}
\newcommand{\bM}{\mathbf{M}}
\newcommand{\bP}{\mathbf{P}}
\newcommand{\bQ}{\mathbf{Q}}

\newcommand{\bT}{\mathbf{T}}
\newcommand{\bA}{\mathbf{A}}
\newcommand{\bK}{\mathbf{K}}
\newcommand{\bI}{\mathbf{I}}

\newcommand{\simplex}{\Delta}

\newcommand{\grad}{\operatorname{grad}}
\newcommand{\dive}{\operatorname{div}}

\usepackage{dsfont}

  \newtheorem{lemma}{Lemma}
  \newtheorem{corollary}{Corollary}
  \newtheorem{proposition}{Proposition}
  \newtheorem{theorem}{Theorem}
  \newtheorem{definition}{Definition}
  \newtheorem{remark}{Remark}
  \newtheorem{claim}{Claim}

 \newenvironment{acknowledgements}{\textbf{Acknowledgements}}{}

\title{
Gradient flow formulations of discrete and continuous evolutionary models: a unifying perspective.
}
\date{\today}
\usepackage{authblk}
\author[1]{Fabio A. C. C. Chalub}
\author[2,3]{L\'eonard Monsaingeon}
\author[1]{Ana Margarida Ribeiro}
\author[4]{Max O. Souza}
\affil[1]{Departamento de Matem\'atica and Centro de Matem\'atica e Aplica\c c\~oes, Faculdade de Ci\^encias e Tecnologia, Universidade Nova de Lisboa, Quinta da Torre, 2829-516 Caparica, Portugal}
\affil[2]{IECL Universit\'e de Lorraine, F-54506 Vand\oe uvre-l\`es-Nancy Cedex, France}
\affil[3]{
Grupo de F\'isica Matem\'atica, GFMUL, Faculdade de Ci\^encias, Universidade de Lisboa,
1749-016 Lisboa, Portugal
}
\affil[4]{Instituto de Matem\'atica e Estatística, Universidade Federal Fluminense, Rua Prof. Marcos Waldemar de Freitas Reis, S/N, Campus do Gragoat\'a, Niter\'oi, RJ 24210-201, Brazil}

\begin{document}

\maketitle

\begin{abstract}
We consider three classical models of biological evolution:
(i) the Moran process, an example of a reducible Markov Chain;
(ii) the Kimura Equation, a particular case of a degenerated Fokker-Planck Diffusion;
(iii) the Replicator Equation, a paradigm in Evolutionary Game Theory.
While these approaches are not completely equivalent, they are intimately connected, since (ii) is the diffusion approximation of (i), and  (iii) is obtained  from (ii) in an appropriate limit.
It is well known that the Replicator Dynamics for two strategies  is a gradient flow with respect to the celebrated Shahshahani distance.
We reformulate the Moran process and the Kimura Equation as gradient flows and in the sequel we discuss conditions such that the associated gradient structures converge: (i) to (ii), and (ii) to (iii).
This provides a geometric characterisation of these evolutionary processes and provides a reformulation of the above examples as time minimisation of free energy functionals.



\end{abstract}

\section{Introduction}

\subsection{Background}
From a contemporary perspective, evolution  can be conveniently described as being the product of changes in allele frequencies within a population --- cf. \cite{ewens2004mathematical}.
Albeit with an apparently simple definition, evolution is actually a complex phenomenon and, as such, it comprises many different mechanisms: natural selection, mutation, genetic-drift, to name only a few. 

The need to understand these different mechanisms and, more recently, their interplay, led to the development of a plethora of models in evolutionary dynamics: discrete time Markov chains were used in the early 20th century to study genetic drift (e.g., the Wright-Fisher~\cite{Wright_1931,Fisher_1922} and the Moran~\cite{Moran_1962} processes);  
continuous time stochastic processes in the mid 20 century geared towards molecular evolution, as is the case of Kimura Partial Differential Equation~\cite{Kimura}; 
and, finally, systems of Ordinary Differential Equations (ODE) that are used to model natural selection, cf.~\cite{TaylorJonker_1978}. 
These three classes of basic models can be considered as a classical triad \cite{ChalubSouza09b,ChalubSouza14a}, and they will be the focus of this work.  It should be also noticed that, more recently, new modelling paradigms have been  used --- most notably 
Individual Based Models
\cite{hauert2004spatial,deangelis2005individual} and kinetic models \cite{bellomo2008mathematical,veloz2014reaction}.  

The study of different connections between  models in this classical triad dates back at least to \cite{EthierKurtz}, where a frequency-dependent version of the Wright-Fisher process was introduced, and the large population regime was shown to be described by a generalised version of the Kimura Equation.
More recently a number of works have explored these links providing various approaches to a unified view on these models, in the weak selection regime with infinite population limit and suitable scaling relations between the time step and population size~\cite{champagnat2006unifying,ChalubSouza09b,ChalubSouza14a}.
In addition, the Kimura Equation and the Replicator Dynamics (RD) are connected for short times and strong selection.
However, despite all these connections, there are also important differences --- see \cite{ChalubSouza16} for results on the  qualitative difference of fixation probability in large populations compared to infinite ones, and  \cite{ChalubSouza_2017} for notable features of the fixation probability in finite populations that are not in the weak selection regime. These results suggest that the impact of all underlying assumptions  in each of these three models is not  yet fully understood.

The aim of this work is to investigate this classical triad  from yet a different perspective which, as far as we are aware of, seems to have been overlooked: namely, the fact that these  models can be formulated  (or at least made compatible with) some sort of local  maximisation principle. 
This approach has a long tradition in the biological literature, cf.~\cite{Kimura_1958,Ewens_1992,Behera_1996}; for the relation between \emph{optimal principles in evolution} and the \emph{Fundamental Theorem of Natural Selection}, due to Fisher~\cite{Fisher_1930}, see also~\cite{Ewens_2011,Ewens_2014,Ewens_2015}.

 We should point out that we are not attempting to obtain a global maximisation principle.
 The existence (or usefulness) of global principles is  a quite controversial topic in evolution, and we refer the reader to  \cite{smith1978optimization} for a review on optimising techniques in evolution and to \cite{metz1996does,metz2008does} for a  critique on this approach.  

\subsection{Gradient flow formulations of evolutionary models and main results}

It will turn out that  the fundamental tool that will allow us to accomplish the previous task is the concept of \emph{gradient flows}. 
This is a rather classical topic in differential equations that has raised much interest recently after the water-shedding work \cite{ambrosio2008gradient} --- see also \cite{santambrogio2017euclidean} for a very gentle introduction.

Gradient flows are hardly new in evolutionary dynamics: under some hypotheses, the RD can be reinterpreted as a gradient flow with respect to a specific metric --- known  by now as the  Shahshahani metric.
In particular, for the one dimensional case, the RD is always a gradient flow in this metric~\cite{Shahshahani_1979,akin1979geometry,akin1990differential}.

Motivated by the positive results of a  research program undertaken  by two of the authors (FACCC and MOS) in studying this triad starting from the discrete processes \cite{ChalubSouza09b,ChalubSouza14a,ChalubSouza16}, we will follow the same pattern with a slight modification: we will start from the continuous-time generalisation of the well-known Moran process \cite{etheridge2011some}.
This will allow us to adapt the framework recently developed by \cite{Maas_JFA} to our case, and obtain a formulation of the transient part of the Moran process as a gradient flow.
These adaptations turn out to be deeply connected with the so-called associated $Q$-process, which  describes the probability law of  eternal paths in this absorbing system.
This formulation will also provide a ``geometrisation"  for finite populations, and answers a question raised in \cite{akin1979geometry}. 

Finally, from well-known gradient flow formulations of Fokker-Planck type equations\cite{Jordan_Kindeleherer_Otto}, and based on the $Q$-process in the continuous setting, we obtain a gradient flow formulation for the Kimura Partial Differential Equation.

Summing up, we were able to derive independent gradient flow formulations of the triad. 
Subsequently, we study the natural compatibility between these models, which is expected to hold due to some of the authors' previous results. 

More precisely, we will:
\begin{enumerate}
	\item
	Reformulate all three models as gradient flows, i.e, in each case we find a free-energy $\mathsf{H}$ defined on a Riemannian manifold with distance $\mathsf{w}$ such that the model is equivalent, in a sense to be made precise, to the steepest descent flow of $\mathsf{H}$ as measured by $\mathsf{w}$. See~\cite{ambrosio2008gradient,santambrogio2015optimal,Villani_Topics,Ambrosio_Gigli_UserGuide} for an overview on the topic.
	\item
	Obtain these processes as a time-step minimisation of $\mathsf{H}$.
	Namely, consider that at time $t$ the system is at state $p$, where $p$ is a probability density distribution describing all possible states of the system. 
	In the next time step, the system will be at state $p'$ such that the infimum of $q\mapsto \frac{h}{2}\mathsf{w}^2(p,q)+\mathsf{H}(q)$ is attained, for a small and positive $h$.
	After the seminal paper by Jordan, Kinderlehrer and Otto~\cite{Jordan_Kindeleherer_Otto}, this second approach became known as JKO scheme. 
\end{enumerate}
Under appropriate convexity assumptions these two approaches are equivalent in a very general setting \cite{ambrosio2008gradient}, but the construction can also be made rigorous in the absence of convexity on a case-by-case basis; 
see, e.g.,~\cite{Laguzet_2018,DiFrancesco_2013,Blanchet_2013,Laurencot_Matioc_2013,Kinderlehrer_Monsaigeon_Xu_2017}. 
The three models are linked by two limiting processes: the partial differential equation (PDE) is connected to the Markov Chain by an infinite population limit, and to the ODE by a vanishing viscosity limit.
One may therefore wonder whether (or hope that) all these gradient flow formulations are compatible in some sense.
An appropriate tool to investigate this turns out to be the $\Gamma$-convergence of gradient flows \cite{serfaty2011gamma,sandier2004gamma} and we will discuss how this connection can be  obtained.

A more thorough discussion of modelling implications to evolution will be postponed to a subsequent work.
However, we should state that both the free energy and metric will be derived from a set of common assumptions in evolutionary dynamics and are not artificial quantities.
The metric is the Wasserstein distance between two probability measures built upon a generalisation of the so-called \emph{Shahshahani metric}~\cite{Shahshahani_1979}, introduced in the framework of gradient systems in the Replicator Equation; cf.~\cite{Harper_2011,Sigmund_1987a,Sigmund_1987b,Hofbauer_1985}, see also the discussion on \emph{Kimura's Maximum Principle} and the (Svirezhev-) Shahshahani metric in~\cite{burger2000mathematical}.
We note also that, for finite populations, short-term and long-term information on the dynamics has been recently obtained from the free energy~\cite{ChalubSouza18}. 

\subsection{Outline}

The structure of this paper is as follows. In the remainder of this section, we first introduce and fix the basic notation that will be used throughout this work and we recall the theory previously established related to the connections between the models.
In the sequel, we review the main results by FACCC, MOS, and collaborators related to the current work.

In Section~\ref{sec:CTDMC}, we introduce a class of matrices that encompass celebrated classes of matrices used in population genetics, and introduce a class of continuous-in-time discrete Markov chains.
In order to reformulate this class of processes as gradient flows, we appeal to the so-called $Q$-processes and generalise recent results  by \cite{Maas_JFA} for irreducible Markov chains to the case with two or more absorbing states. 

In Section~\ref{sec:DTDMC}, we digress and discuss the relation between the discrete and continuous time Markov chains.
While the discrete ones are not amenable to the treatment developed in Section~\ref{sec:CTDMC}, we opt to include it here as it is the most traditional framework for finite-population models in population genetics.
Furthermore, it has been the starting point of previous work from part of the authors, as described above.
In particular, we show how the results from the previous section can be used to define entropies in population dynamics.
The  framework developed herein is sufficiently general to treat simultaneously the Boltzmann-Gibbs-Shannon (BGS) entropy and the Tsallis entropies (both in discrete and continuous time).

In Section~\ref{sec:CTCMC}, we describe the gradient flow formulation of the Kimura Equation, a degenerated PDE of drift-diffusion type. As already shown in \cite{Chalub_Souza:CMS_2009,ChalubSouza14a} the appropriate formulation satisfies two conservation laws and has measure-valued solutions. However, as in the discrete case, it turns out that only the interior dynamic can be formulated as a gradient flow and we focus on this part of the solutions.

In Section~\ref{sec:replicator} we study the Replicator Equation, which is well known \cite{Shahshahani_1979,akin1979geometry} to have a variational structure.
In fact, we shall rather study its fully equivalent PDE version.
The latter is more appropriate for our framework, but all our results immediately apply to the former.

In Section~\ref{sec:gfcmp}, we discuss the compatibility between the variational structure (both gradient flow and JKO formalism) for all the models discussed in Sections~\ref{sec:CTDMC},~\ref{sec:CTCMC}, and~\ref{sec:replicator}.
As mentioned earlier the $\Gamma$-convergence of gradient flows will be the important tool in this section --- see \cite{dal2012introduction} for a general introduction  on $\Gamma$-convergence and \cite{serfaty2011gamma,sandier2004gamma} for $\Gamma$-convergence of gradient flows --- see also \cite{gigli2013gromov}.

We finish with some comments in Section~\ref{sec:conclusion}.

A roadmap of the paper summarising the main relationships between the three models is given in Figure~\ref{fig:outline}.

\begin{figure}
\includegraphics[width=\textwidth]{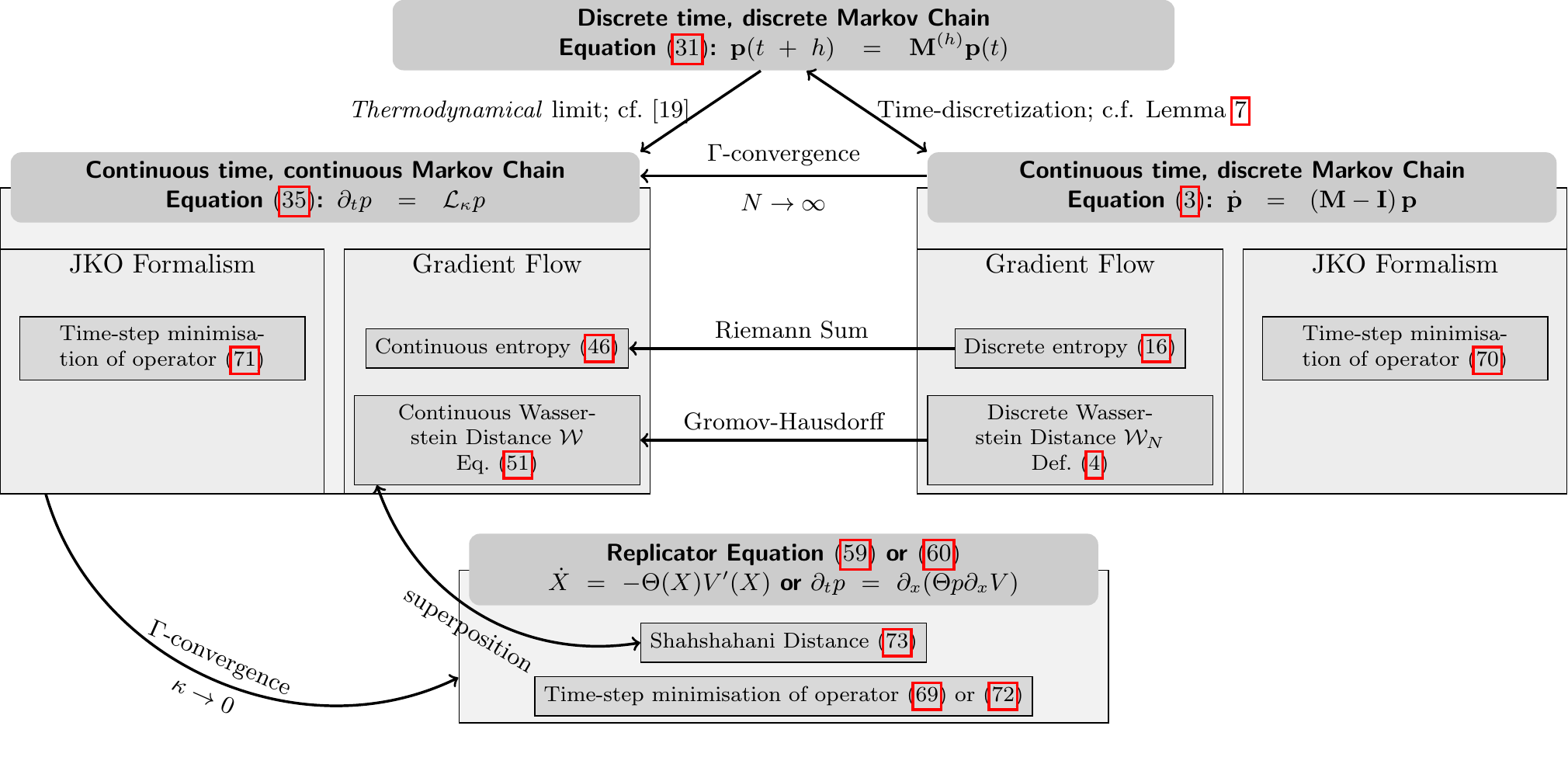}
\caption{Roadmap to the paper.} 
\label{fig:outline}
\end{figure}

\subsection{Notations}

Let us consider a $N+1$ Markov chain on an abstract finite space $\mathcal X_N$.
By abuse of notation, we write $\mathcal X_N\bydef\left\{0,\dots,N\right\}$.
We will use bold symbols for discrete vectors and matrices.
Vectors are considered as column vectors by default.
Given a matrix, and unless otherwise specified, the words \emph{stochastic} and \emph{substochastic} mean \emph{column-stochastic} and \emph{column-substochastic}, respectively.

By $\mathcal{P}(\mathcal X_N)$ we denote the space of probability measures in $\mathcal X_N$, canonically identified with vectors in the $N$-dimensional simplex $\simplex^N\bydef\{\bp\in\R_+^{N+1}|\sum_ip_i=1\}$. 
We shall use lowercase to denote probability vectors, and uppercase for their densities with respect to some reference measure: Typically $\bp=\left(p_i\right)_{i\in\mathcal X_N}$ will denote an arbitrary probability vector; if $\bpi=\left(\pi_i\right)_{i\in\mathcal X_N}$ is a particular reference probability measure, we write $\bP\bydef\frac{\rd\bp}{\rd\bpi}=\left(\frac{p_i}{\pi_i}\right)_{i\in\mathcal X_N}$ for the density of the measure $\bp$ with respect to the measure $\bpi$.

For Markov chains with more than one absorbing states, transient parts are particularly relevant.
Therefore, we denote them by tilded quantities $\widetilde\bp=(p_i)_{i=0,\dots, N-k}$, where $k$ is the number of linearly independent absorbing states.
In a slightly inconsistent notation, and whenever we focus on chains with only two absorbing states, these will sometimes be labelled instead as $i=0$ and $i=N$ for convenience, and tilded quantities will denote projections to the ``interior'', e.g. $\tilde\bp=(p_i)_{i=1,\dots, N-1}$.
In simple words, we discard the $i=0, N$ entries in the vector $\bp=(p_i)_{i=0,\dots, N}$.

The transition matrix of the Markov chain on the state space $\mathcal X_N$ is given by a $(N+1)\times(N+1)$ matrix $\bM=\left(M_{ij}\right)_{i,j\in\mathcal X_N}$.
The matrix $\widetilde{\bM}\bydef\left(M_{ij}\right)_{i,j=0,\dots,N-k}$ associated to the transient part of the process is called the \emph{core} matrix associated to $\bM$.
(Depending on the particular labelling, we also consider $\widetilde{\bM}=\left(M_{ij}\right)_{i,j=1,\dots,N-1}$.)

If $\mathbf{A},\mathbf{B}\in\R^{n\times m}$, then we will write $\mathbf{A}\circ\mathbf{B}\in\R^{n\times m}$ to denote their entry-wise (Hadamard) product --- i.e. $\left(\mathbf{A}\circ\mathbf{B}\right)_{ij}=A_{ij}B_{ij}$. Note that $\left(\mathbf{A}\circ\mathbf{B}\right)^\dagger=\mathbf{A}^\dagger\circ\mathbf{B}^\dagger$. For $\mathbf{a}\in\R^n$, we will write  $\diag{\mathbf{a}}\in\R^{n\times n}$ to denote the $n\times n$  matrix with its main diagonal given by $\mathbf{a}$ and zero elsewhere. Recall that the usual matrix product and the Hadamard product agree for  diagonal matrices --- in particular, $\diag{\mathbf{a}\circ\mathbf{b}}=\diag{\mathbf{a}}\diag{\mathbf{b}}$, and $\frac{\rd\bp}{\rd\bpi}=\bp\circ\bpi^{-1}$.

In the continuous setting, $x\in\Omega\bydef[0,1]$ denotes the state of the system, i.e., $x$ is the fraction of individuals of the focal type. 
In particular, the absorbing states are indicated by $x\in\partial\Omega\bydef\{0,1\}$. This explains the reordering in the indexes used in the discrete setting when there are exactly 2 absorbing states. We shall write $p\in \mathcal P(\Omega)$ for measures on the whole domain, while $\widetilde p$ will denote the restriction $p\measurerestr{\Omega}$ to the interior, i.e. $p=a\delta_0 +\widetilde p + b\delta_1$, for $a,b\ge0$. We denote by $q\ll p$ and $q\perp p$ if $q$ is absolutely continuous with respect to $p$, and if $q$ is singular with respect to $p$, respectively.

In general, tilded quantities, both in the discrete and continuous settings, are not probability measures and we will often need to rescale and/or renormalise those interior projections so that they become again probability measures: the resulting scaled variables will be denoted by $\bq=(q_i)_{i=1\dots N-1}$ or $q\in\mathcal P(\Omega)$, while we keep the letter $\bp,p$ for the initial variables.

Unless otherwise specified dotted quantities will denote time derivatives $\dot X(t)=\frac{\rd X}{\rd t}(t)$, while primes will stand for spatial derivatives $V'(x)=\frac{\rd V}{\rd x}(x)$. 

A time step will be denoted $\Delta t$ if it is a parameter of the model, or $h$ if it corresponds to a discretization of the continuous time variable $t$.
\subsection{State of the art}\label{subsec:state of art}

We finish this introduction with a more detailed explanation of the triad, and the links between these three different classes. Note that in part of the current work, however, we opt to describe more general processes than the one exemplified here: in Sections~\ref{sec:CTDMC} and~\ref{sec:DTDMC}, we consider stochastic processes with an arbitrary number $k>1$ of absorbing states, while in Sections~\ref{sec:CTCMC} and~\ref{sec:replicator}, we consider more general diffusion coefficients than in this subsection.

The Moran process was introduced in~\cite{Moran_1962} as a mathematical simplification of the older Wright-Fisher process~\cite{Wright_1931,Fisher_1922} and it is a particular example of a birth and death process.
In the setting we are interested in, we consider a population of fixed size $N$ divided into two groups of individuals, indistinguishable apart from the characteristic under study.
Let us call these two \emph{types} \A and \B.
Every $\Delta t$ seconds, one individual is chosen to die with uniform probability $1/N$, while a second one (possibly the same one) is chosen to reproduce according to a certain \emph{type selection probability vector} $(\tsp_0,\tsp_1,\dots,\tsp_N)$.
Here $\tsp_i$ indicates the probability to select for reproduction an individual of type \A in a population with $i$ individuals of type \A.
In this case, we say that the population is at state $i$. 
Because the Moran process has no mutations (i.e., $\tsp_0=1-\tsp_N=0$) there are two absorbing states, namely $i=0$ and $i=N$.

At each time step the transition probability from state $j$ to $i$ is thus given by a $(N+1)\times(N+1)$ matrix $\bM=\left(M_{ij}\right)_{i,j=0,\dots,N}$ with
\begin{equation}\label{def:Moran_TM}
M_{ij}=\left\{
\begin{array}{ll}
\frac{N-j}{N}\tsp_j,&\ i=j+1\ ,\\
\frac{j}{N}\tsp_j+\frac{(N-j)}{N}(1-\tsp_j),&\ i=j,\\
\frac{j}{N}(1-\tsp_j),&\ i=j-1,\\
0,&\text{otherwise.}
\end{array}
\right.
\end{equation}
The evolution equation (also known as \emph{master equation}) is given by 
\begin{equation}\label{DTDMC_ee}
\bp(t+\Delta t)=\bM\bp(t)\ ,
\end{equation}
where $\bp(t)=(p_0(t),\dots,p_N(t))$, and $p_i(t)$ indicates the probability to find the system at state $i$ at time $t$.

An alternative view of the Moran process will be presented in Section~\ref{sec:CTDMC}, where time will be considered a continuous variable and therefore, the evolution equation will rather be given by
\begin{equation}\label{CTDMC_ee}
\frac{\rd\bp}{\rd t}=\left(\bM-\bI\right)\bp,
\end{equation}
being $\bI$ the $(N+1)\times(N+1)$ identity matrix.
The obvious link between equations~\eqref{DTDMC_ee} and~\eqref{CTDMC_ee} will be fully exploited in Section~\ref{sec:DTDMC} and will be instrumental to build the link between finite and infinite population evolutionary models, and also in order to translate to the more usual setting \eqref{DTDMC_ee} all the results found for \eqref{CTDMC_ee}.
In fact, as we will see in a moment, infinite population in previous works of some of the authors is derived from the discrete-time evolution \eqref{DTDMC_ee}, while the gradient flow formulation, the main object of the present work, will require from the start a continuous time.
By contrast with previous works, the infinite population model will be derived here from the continuous time \eqref{CTDMC_ee}. 

In order to study the limit $N\to\infty$, $\Delta t\to 0$ it is necessary to assume a certain scaling relation between the population size and time-step, as well as the so-called \emph{weak selection principle}: At leading order, the type selection probability must be of the specific form
\begin{equation}\label{eq:wsp}
\tsp_i=\frac{i}{N}\left[1-\frac{2}{\kappa N}\left(1-\frac{i}{N}\right)V'\left(\frac{i}{N}\right)\right]
\end{equation}
for a given potential $V(x)$ (the gradient $-V'$ representing the fitness difference between the focal type, \A, and its opponent \B).
The parameter $\kappa>0$ is the inverse of the selection strength, see~\cite{ChalubSouza18} for a detailed analysis of each parameter in Equation~\eqref{eq:wsp}.
In fact, if $N\left(\Delta t\right)^{1/2}\to\kappa^{-1}$, then the vector $\bp$ obtained from Equation~\eqref{DTDMC_ee}, given a certain initial condition, converges in an appropriate sense to a measure $p$, where $p$ is the solution of a certain degenerate parabolic partial differential equation of drift-diffusion type known as \emph{the Kimura equation}
\begin{equation}\label{eq:KimuraPDE}
\partial_tp=\frac{\kappa}{2}\partial_{xx}^2\left(x(1-x)p\right)+\partial_x\left(x(1-x)V'(x)p\right)
\end{equation}
for $(x,t)\in[0,1]\times \R_+$.
This equation must be supplemented by two conservation laws
\[
\frac{\rd\ }{\rd t}\int_0^1p(t,x)\,\rd x=0
\qquad\mbox{and}\qquad
\frac{\rd}{\rd t}\int_0^1\varphi(x) p(t,x)\,\rd x=0,
\]
where $\varphi$ is the unique solution of $\varphi''-V'(x)\varphi'=0$ with boundary conditions $\varphi(0)=0$, $\varphi(1)=1$.
The initial condition will be the limit, in the same sense, of the initial conditions of the discrete process. See~\cite{ChalubSouza09b} for the derivation of Equation~\eqref{eq:KimuraPDE} from the Moran process, and~\cite{ChalubSouza14a} for its generalisation to an arbitrary number of types derived from the Wright-Fisher process, i.e. a process such that the transition matrix probability from state $j$ to $i$ is given by $M_{ij}=\binom{N}{i}\tsp_j^i(1-\tsp_j)^{N-i}$.
Finally, see~\cite{Chalub_Souza:CMS_2009,Danilkina_etal} for the detailed study of the Kimura Equation~\eqref{eq:KimuraPDE} from the partial differential equation point of view.

As a last remark, we note that when $\kappa\ll1$, the limit of the Kimura Equation is the transport equation $\partial_tp=\partial_x\left(x(1-x)V'(x)p\right)$, which is a PDE version of the well-known \emph{Replicator Dynamics}
\begin{equation}\label{eq:replicador}
\dot X=-X(1-X)V'(X)\ . 
\end{equation}

\section{Continuous in time, discrete Markov chains}\label{sec:CTDMC}

In this section, we consider continuous time Markov chains on $\mathcal X_N$, given by Equation~\eqref{CTDMC_ee} with $\bM$ a given stochastic matrix, and initial condition $\bp(0)=\bp^\ini$.
When $\bM$ is irreducible, $p(t)$ converges as $t\to\infty$ to a unique and strictly positive invariant probability measure~\cite{karlin1981second}.
Under the additional assumption that $\bM$ is reversible (see Subsection~\ref{subsec:micro_reversible}), one can construct as in \cite{Maas_JFA} a discrete Wasserstein distance $\mathcal{W}_N$ on the space of probabilities $\mathcal P(\mathcal X_N)$ such that~\eqref{CTDMC_ee} is the gradient flow of the relative Boltzmann-Shannon-Gibbs entropy, to be defined at Subsection~\ref{ssec:Grad_Flow_form}.

The same discrete distance was constructed independently in \cite{mielke2013geodesic,chow2012fokker,Maas_JFA}
and will be discussed later on in Subsection~\ref{sec:gradient flow}.
Therefore, we refrain from giving the details and precise definitions at this early stage, but it is worth pointing out that this theory of \emph{discrete optimal transport} crucially requires irreducibility \emph{and} reversibility of the Markov kernel $\bM$.
The non-irreducibility of a Markov process is typically due to the existence of absorbing states.

Our goal here is to obtain the aforementioned variational framework for a class of \emph{reducible} Markov chains. In particular we aim at providing a gradient flow structure for some models  of population dynamics that are not covered by a straightforward application of the results in~\cite{Maas_JFA}, yet include the Moran process.

With that goal in mind, we will first rephrase the aforementioned results to substochastic, irreducible and reversible chains, and subsequently apply the results to our chains of interest, introduced in Subsection~\ref{subsec:admissible_mat}.

Roughly speaking we shall focus on Markov processes for which a particular subdynamics can be identified and allows to reconstruct the whole dynamics, and such that the subdynamics can be recast into an irreducible, reversible Markov process.
As alluded to in the introduction, this subdynamics is the core dynamics and corresponds to the evolution of the transient states only.
We first make these structural assumptions on the Markov kernels $\bM$ more precise and technically explicit in terms of linear algebra (Subsections~\ref{subsec:admissible_mat} and~\ref{subsec:micro_reversible}).
We discuss next the relation between those technical assumptions and probabilistic conditioning of the original processes (Subsection~\ref{subsec:Q_process}), we discuss the resulting variational framework (Subsection~\ref{sec:gradient flow}), and finally we apply this framework to a time continuous version of one of the models of triad~(Subsection~\ref{ssec:moran_proc}).

\subsection{Admissible matrices}
\label{subsec:admissible_mat}

In the sequel we will consider  Markov processes with  $N+1$ states and  $1\leq k < N-1$  absorbing states and such that the chain conditioned on non-absorption is irreducible.
Moreover, we assume that all absorbing states are accessible.
More explicitly,

\begin{definition}[Admissible chains]
	\label{def:ac}
	Let $1\leq k<N-1$.
	We say that a $(N+1)\times (N+1)$ stochastic matrix  $\bM$ is $(N+1,k)$-admissible (admissible, in short), denoted $\bM\in \calA_{N+1,k}$, if there exists a permutation matrix $\bR$ such that
	\begin{equation}
	\bR^\dagger\bM\bR=\begin{pmatrix}
	\widetilde{\bM}&\mathbf{0}\\
	\bA&\bI
	\end{pmatrix}
	\label{eq:def_admissible}
	\end{equation}
	where $\bI$ is the $k\times k$ identity, $\widetilde{\bM}$ is an $(N+1-k)\times(N+1-k)$ irreducible matrix, and no row of $\bA$ is identically zero.
	We will refer to $\widetilde{\bM}$ as the core matrix associated to $\bM$. 
\end{definition}
One should think here of $k$ absorbing states that do not interact with each other and ${(N+1)-k}$ transient states that possibly self-interact or lose information by getting absorbed.
These features are encoded by the matrices $\bI$, $\widetilde\bM$, and $\bA$, respectively.

Observe that our structural normalisation~\eqref{eq:def_admissible} is not unique, since one can always further permute any of the first $(N+1-k)$ columns and rows (corresponding to relabelling the transient states) while keeping a similar structure.
By abuse of notations we will still talk of \emph{the} core matrix $\widetilde\bM$, which is thus defined only up to permutations.
As a consequence we always think of the permutation matrix $\bR$ as the identity matrix, and of the Markov kernel as already in the canonical form
\[
\bM=\begin{pmatrix}
	\widetilde{\bM}&\mathbf{0}\\
	\bA&\bI
	\end{pmatrix}
\]

In what follows, we shall denote the dynamics of the $(N+1-k)$ transient states as
\begin{equation}
\label{eq:evolution_p_tilde}
\widetilde \bp=(p_i)_{i=0,\dots, N-k},
\qquad
\frac{\rd\widetilde\bp}{\rd t}=(\widetilde\bM-\bI)\widetilde\bp.
\end{equation}
By definition~\ref{def:ac} of admissible chains the absorption matrix $\bA$ has non-zero rows and $\widetilde \bM$ is strictly \emph{sub}stochastic: $\sum_i\widetilde p_i(t)$ is therefore non-increasing in time, and the transient states leak information to the absorbed states.

\begin{remark}[Kimura matrices]
	\label{rmk:kimura}
	The class of $(N+1,k)$-admissible matrices is an extension of a number of classes previously investigated.
	In particular, $\calA_{N+1,2}$ denotes the so-called \emph{Kimura matrices}~\cite{ChalubSouza_2017}, which is relevant to evolutionary dynamics.
	For this class, which includes the Moran and Wright-Fisher processes discussed in the introduction, a different presentation and notation were used: the matrix is naturally ordered with $i$ indicating the presence of a focal type, and  the two absorbing states are labelled $i=0,N$ with
	\begin{equation}
	\label{eq:pres_kimura}
    \bM=\begin{pmatrix}1&\boldsymbol\alpha^\dagger&0\\
	\mathbf{0}&\widetilde{\bM}&\mathbf{0}\\
	0&\boldsymbol\beta^\dagger&1
	\end{pmatrix}
	\end{equation}
	Here $\boldsymbol\alpha$, $\boldsymbol\beta$ and $\bzero$ are $(N-1)\times 1$ vectors, with $\boldsymbol\alpha$ and $\boldsymbol\beta$ being  nonnegative and nonzero.
\end{remark}

	By stochasticity of $\bM$ we always have as an eigenvector $\bone\bydef(1,\dots,1)$, and we can choose a basis  $\lbrace\bF_1,\ldots,\bF_k\rbrace$ of the left-eigenspace of $\bM$ (which is indeed $k$-dimensional from Definition~\ref{def:ac}) comprised of non-negative vectors such that $\sum_{j=1}^k\bF_j=\bone$.
	One readily checks that solutions to~\eqref{CTDMC_ee} automatically satisfy the conservation laws
	\begin{equation}
	\label{eq:discrete_conservation_laws}
	\langle \bp(t),\bF_l\rangle=	\langle \bp(0),\bF_l\rangle,
	\qquad \forall \, t>0\text{ and }l=1,\ldots,k.
	\end{equation}
	We refer to ~\cite{ChalubSouza09b,ChalubSouza14a} for a discussion on how these conservation laws matter for the dynamics when considering the diffusive (continuous) approximation of Markov chains with absorbing states.
\begin{remark}
\label{rmk:discrete_fixation}
	\label{rmk:cons-law}
	In the particular case of Kimura matrices $\bM\in\calA_{N+1,2}$, the distinguished left eigenvectors are taken to be $\bone$ and $\bF$, with the first and the last entries of $\bF$ being zero and one according to the representation of Equation~\eqref{eq:pres_kimura}.
	The vector $\bF$ is the \emph{fixation probability} of the focal type  --- see \cite{ChalubSouza_2017} and references therein and Remark~\ref{rmk:fixation_continuous} for the continuous version.
\end{remark}

\begin{remark}[Multi-type admissible processes]
	\label{rmk:other_adm}
	As we shall see below, the Moran process with $k$-types is not admissible because the  matrix $\widetilde{M}$ is not irreducible.
	This is a consequence of non-interaction requirements  that must be satisfied by the inner dynamics (when a type is extinct,  it cannot reappear).
	On the other hand, it is possible to construct a birth-death process with $k$ absorbing states, such that $\widetilde{M}$ is irreducible.
	As an example, we mention a process  akin to the Moran process with \emph{frequency dependent mutations}.
	The easiest example of such a process can be briefly defined as follows: consider  this process with a three-type population and set to zero all transition probabilities from the homogeneous states to any other state --- this modification will still allow mutations from states with two types into states with three types.
	Such a process will have an irreducible $\widetilde{M}$ and three absorbing states. Naturally, this can be extended to $k$-type process, with $k$ absorbing states. 
\end{remark}

It will be also convenient to make explicit the structure of the semigroup associated to the forward Equation~\eqref{CTDMC_ee}:

\begin{lemma}
	\label{lem:expl_sg}
	Let $\bM$ be admissible.
	Then the fundamental solution to~\eqref{CTDMC_ee} is given by
    \[
	\exp\left(t(\bM-\bI)\right)=\begin{pmatrix}
	\exp\left(t(\widetilde{\bM}-\bI)\right)&\bzero\\
	\bA(\widetilde{\bM}-\bI)^{-1}\left(\exp\left(t(\widetilde{\bM}-\bI)\right)-\bI\right)&\bI
	\end{pmatrix},
	\]
	where we write indistinctly $\bI$ for the identity matrix of dimension either $N+1$ or $N+1-k$.
\end{lemma}
\begin{proof}
 This follows from~\eqref{eq:def_admissible}.
 \qed
\end{proof}

As already discussed $\widetilde{\bM}$ is strictly substochastic -- the rows of $\bA$ being non zero in~\eqref{eq:def_admissible} -- and therefore its spectral radius
\[
\mu\bydef\rho(\widetilde{\bM})\in(0,1).
\]
The Perron-Frobenius theorem implies next that $\mu$ is the dominant eigenvalue of $\widetilde\bM$, and both its associated left and right eigenvectors can be chosen positive, c.f.~\cite{horn1990matrix}.
Following up on the rough idea that the transient dynamics determines the whole evolution, we define next
\begin{definition}[Characteristic triple]
\label{def:characteristic_triple}
	 Let $\bM$ be admissible with core matrix $\widetilde{\bM}$, and $\mu=\rho(\widetilde{\bM})$.
	 In addition, let  $\widetilde{\bw}$ and $\widetilde{\bz}$ be the unique positive left and right eigenvectors associated to $\mu$, normalised as $\langle\widetilde{\bw},\widetilde{\bone}\rangle=\langle\widetilde{\bw},\widetilde{\bz}\rangle=1$.
	 We will term $(\mu,\widetilde{\bw},\widetilde{\bz})$ the \emph{characteristic triple}.
\end{definition}

In general, it is difficult to derive explicit expressions for either $\mu$  or $\widetilde{\bw},\widetilde{\bz}$.
Nevertheless, formulas can be obtained in certain particular cases or asymptotically in the limit of large populations $N\to\infty$, see Subsections~\ref{ssec:moran_proc} and~\ref{ssec:kimura_eq}.

\subsection{Micro-reversible processes}
\label{subsec:micro_reversible}
In order to exploit the results from \cite{Maas_JFA}, we need to restrict ourselves to processes that satisfy some reversibility, at least to some extent.
We introduce below a generalised notion of reversibility, adapted for the case of substochastic matrices, that we shall call \emph{micro-reversibility}.
Intuitively, micro-reversible processes should be time-reversible in the meta-stable regime.
	More precisely, micro-reversibility means that, when considering the difference between the mass flow from $i$ to $j$ and the mass flow from $j$ to $i$, the total loss of mass at site $i$ is independent of both $i$ and $j$, provided neither $i$ or $j$ are absorbing sates:
\begin{definition}[Micro-reversible matrices]
	\label{def:mr_mat}
	Let $\bM$ be an admissible matrix with core $\widetilde \bM$ in the sense of Definition~\ref{def:ac} and with characteristic triple $(\mu,\widetilde{\bw},\widetilde{\bz})$ as in Definition~\ref{def:characteristic_triple}.
	We say that $\bM$ is \emph{micro-reversible} if 
	\begin{equation}
	\label{eq:generalized_detailed_balance}
	\widetilde{w}_i\widetilde{M}_{ij}\widetilde{z}_j=\widetilde{w}_j\widetilde{M}_{ji}\widetilde{z}_i
	,\quad \forall\, i,j=0,\ldots,N-k.
	\end{equation}
	In matrix notation this amounts to requiring symmetry of $\diag{\widetilde{\bw}}\widetilde{\bM}\diag{\widetilde{\bz}}$. 
\end{definition}
\noindent
In a certain sense, this means that for these \emph{slow} processes, a strong equilibrium relation is valid at each step, which resembles the \emph{quasi-stationary} or \emph{ergodic} processes in physics.

Note that this definition generalises the usual notion of reversibility for irreducible stochastic processes.
Indeed for irreducible and column-stochastic Markov chains we have by definition $\widetilde \bM=\bM$, the left leading eigenvector is $\widetilde\bw=\bw=\bone$, and our condition~\eqref{eq:generalized_detailed_balance} reduces to the usual reversibility (\emph{detailed balance}) $M_{ij}z_j=M_{ji}z_i$, $\forall i,j$.
For irreducible row-stochastic matrices the right-eigenvector $\bz=\bone$, and reversibility reads instead $w_iM_{ij}=w_jM_{ji}$.
In these cases, we say that $\bM$ is column- or row-reversible, respectively.
When no confusion arises we simply say that $\mathbf{M}$ is reversible, cf.~\cite{kelly2011reversibility,mielke2013geodesic}.

As we will see in a moment, micro-reversibility is satisfied at least for a particular class of processes, the birth-death processes with two absorbing states.
This includes the Moran process (but not the Wright-Fisher one).
For irreducible chains, birth-death processes with two absorbing states are among the simplest examples of reversibility, cf. \cite{kelly2011reversibility}.
For micro-reversibility, it is not difficult to prove that
\begin{lemma}\label{lem:tridiagonal_microreversible}
 Let $\bM\in\mathcal{A}_{N+1,2}$ such that $\widetilde\bM$ is tridiagonal. Then, $\bM$ is micro-reversible.
\end{lemma}
Note that this is not true in general for admissible matrices with $k>2$ absorbing states, even if $\widetilde{\bM}$ is tridiagonal.
\begin{proof}
Let $(\mu,\widetilde{\bw},\widetilde{\bz})$ be the characteristic triple of $\bM$.
By assumption and up to permutation if needed, the core $\widetilde \bM$ is irreducible and tridiagonal, hence
from standard linear algebra \cite{horn1990matrix} there exists a positive vector $\mathbf{d}$ such that $\bT=\diag{\mathbf{d}}\widetilde{\bM}\diag{\mathbf{d}^{-1}}$ is symmetric, i.e, $\widetilde{\mathbf{w}}\circ\mathbf{d}^{-1}=\mathbf{d}\circ\widetilde{\mathbf{z}}$.
From the identity 
\[
\diag{\widetilde{\mathbf{w}}}\widetilde{\bM}\diag{\widetilde{\mathbf{z}}}
 =\diag{\widetilde{\mathbf{w}}\circ\mathbf{d}^{-1}}\bT\diag{\mathbf{d}\circ\widetilde{\mathbf{z}}}\ ,
 \]
and the symmetry of $\bT$,
 the micro-reversibility condition follows immediately. \qed
\end{proof}

\subsection{The associated $Q$-process}
\label{subsec:Q_process}

Given an absorbing Markov process $X_t$, the associated $Q$-process consists in conditioning it to non-absorption.
Roughly speaking, the corresponding law at a fixed time $s\geq0$ is given by $\lim_{T\to \infty}\mathbb{P}\left[X_s=x\,|\,\tau>T \right]$, where $\tau$ is the absorbing stopping time --- cf. \cite{lamperti1968conditioned,athreya1972branching,cattiaux2009quasi,cattiaux2010competitive,cmm2013}.
In our current setting, the importance of these processes is twofold:
(i) there is a one-to-one correspondence between the transition matrix of the original and conditioned processes;
(ii) the conditioned chain fits the framework of~\cite{Maas_JFA}.
We refrain from going into the technical details and give instead a more direct definition in terms of linear algebra.
The interested reader can check that the evolution of the above limiting process is indeed given by the transition matrix $\bK$ below:

\begin{lemma}[$Q$-process kernels]
 \label{lem:Qproc_is_reversible}
	Let $\bM\in\calA_{N+1,k}$ with characteristic triple $(\mu,\widetilde{\bw},\widetilde{\bz})$.
	The associated $Q$-process is defined by its $(N+1-k)\times(N+1-k)$ transition matrix
	\begin{equation}
	\label{eq:def_K}
	\bK:=\frac 1\mu \diag{\widetilde{\bw}}^{-1}\widetilde{\bM}^\dagger\diag{\widetilde{\bw}}.
	\end{equation}
The kernel $\bK$ is irreducible and row-stochastic, and its unique positive stationary probability distribution is given by 
\begin{equation}
\label{eq:def_pi_discrete}
\bpi:=\widetilde{\bw}\circ\widetilde{\bz} .
\end{equation}
Furthermore, if $\bM$ is micro-reversible then $\bK$ is row-reversible.
\end{lemma}

\begin{proof}
	Since $\widetilde\bw$ is strictly positive and $\widetilde \bM$ is irreducible, clearly $\bK$ is irreducible.
	First, we check that $\bpi$ defined by \eqref{eq:def_pi_discrete} is indeed a left-eigenvector:
	\begin{align*}
	\bpi^\dagger\bK
	&=(\widetilde{\bz}\circ \widetilde{\bw})^\dagger\, \mu^{-1}\diag{\widetilde{\bw}}^{-1}\widetilde{\bM}^\dagger\diag{\widetilde{\bw}}
	\\
	&=\mu^{-1}\widetilde{\bz}^\dagger\widetilde{\bM}^\dagger\diag{\widetilde{\bw}}
	=\widetilde{\bz}^\dagger\diag{\widetilde{\bw}}
	=\bpi^\dagger.
	\end{align*}
	By standard Perron-Frobenius theory $\bpi$ is thus the unique dominant eigenvector, and positive.
From definition~\ref{def:characteristic_triple} we see that $\bpi$ is correctly normalised to be a probability vector.
The fact that $\bK$ is row-stochastic follows from $\bK\bone=\bone$.
The row-reversibility of $\bK$, i.e. $\pi_iK_{ij}=\pi_jK_{ji}$ for all $i,j$ follows from the definition of $\bpi$ and $\bK$, and from the micro-reversibility of $\bM$, Equation~\eqref{eq:generalized_detailed_balance}. \qed
\end{proof}
As already discussed, the transient dynamics leaks mass to the absorbed states.
More explicitly, from~\eqref{eq:evolution_p_tilde} and because $\widetilde\bw$ is a left-eigenvector, we have
\begin{equation*}
\frac{\rd}{\rd t}\langle\widetilde\bw,\widetilde\bp(t)\rangle
=\langle \widetilde \bw,(\widetilde\bM-\widetilde\bI)\widetilde\bp(t)\rangle
=(\mu-1)\langle\widetilde\bw,\widetilde\bp(t)\rangle
\end{equation*}
which implies
\begin{equation}
\label{eq:expo_decay_p}
\langle\widetilde\bw,\widetilde\bp(t)\rangle = e^{(\mu-1)t}\langle\widetilde\bw,\widetilde\bp(0)\rangle\ .
\end{equation}
In particular, since $\widetilde \bM$ is substochastic with spectral radius $\mu\in(0,1)$ the $\widetilde\bw$-weighted mass $\langle\widetilde\bw,\widetilde\bp(t)\rangle$ of the transient states decays at an exponential rate $\mu-1<0$.
Moreover, since $\widetilde w_i>0$ we see that either the initial data is completely absorbed $\widetilde{\bp}(0)=\bzero$ and the dynamics trivially remains absorbed $\widetilde{\bp}(t)\equiv \bzero$, or $\widetilde{\bp}(0)\neq \bzero$ and therefore $\widetilde{\bp}(t)\neq \bzero$ for all $t>0$.
We therefore discard the trivial case $\widetilde{\bp}(0)=\bzero$, and thus we can assume that $\langle\widetilde\bw,\widetilde\bp(t)\rangle\neq 0$ for all times.
We then define the two new $(N+1-k)$-dimensional variables 
\begin{equation}\label{eq:rescaling}
 \bq(t)\bydef\frac{\widetilde \bw\circ\widetilde\bp(\sfrac{t}{\mu})}{\langle\widetilde \bw,\widetilde\bp(\sfrac{t}{\mu})\rangle}
 \quad\text{and}\quad
\bQ(t)\bydef\frac{\rd \bq}{\rd\bpi}(t)
=\left(\dfrac{q_i(t)}{\pi_i}\right)_i.
\end{equation}

By definition of $\bq$ it is clear that $\sum_i q_i=1$.
We think here of $\bpi=\widetilde\bw\circ\widetilde\bz$ as the new reference stationary measure, of $\bq$ as a new probability evolving on a new time-scale $\sfrac t\mu$, and $\bQ=\frac{d\bq}{d\bpi}$ as the density of $\bq$ with respect to $\bpi$.
J. Maas' theory \cite{Maas_JFA} of discrete Wasserstein distances rather take the \emph{density} $\bQ$ as a primary variable, while the $\bq$ probability variable will be more convenient to address the diffusive limit of large populations later on.
Direct substitution and elementary matrix algebra yields:
\begin{lemma}\label{lem:three_equivalent}
	The following three dynamics are equivalent:
	\begin{enumerate}
		\item $\dfrac{\rd \widetilde{\bp}}{\rd t} = \left(\widetilde{\bM}-\bI\right)\widetilde{\bp}$;
		\item $\dfrac{\rd\bq}{\rd t} = \left(\bK^\dagger-\bI\right)\bq$;
		\item $\dfrac{\rd\bQ }{\rd t}= \left(\bK-\bI\right)\bQ$\ .
	\end{enumerate}
In addition, $\langle \bQ(t),\bpi\rangle =\langle \bq(t),\bone\rangle =1$.
\end{lemma}
\noindent
For the sake of brevity we omit the (elementary) proof.

It is worth pointing out that the change of timescale $\sfrac t\mu$ in \eqref{eq:rescaling} is needed for notational convenience only, otherwise an additional factor $\mu$ would appear in the evolution laws below.
Also, in the limit of large populations $N\to\infty$ considered later on, the subdominant eigenvalue $\mu=\mu_N\to 1$ so this rescaling becomes irrelevant.

\subsection{Gradient flow formulation}\label{sec:gradient flow}

In the previous section, we gave a canonical construction of the irreducible, reversible $Q$-process starting from the initial reducible, irreversible process.
With irreducibility and reversibility newly satisfied by the transition matrix $\bK$ of this $Q$-process (Lemma~\ref{lem:Qproc_is_reversible}), we can now apply Maas' theory \cite{Maas_JFA} and identify the $\bq$ evolution as a gradient flow for some discrete optimal transport structure (to be recalled in a moment).

Given an irreducible, reversible Markov kernel $\bK$ (indexed as above by $i=0\dots N-k$) and its unique stationary distribution $\bpi$, the BGS entropy, also known as Kullback-Leibler divergence, of a probability $\bq$ computed relatively to $\bpi$, is defined as
\begin{equation}
\label{eq:def_entropy_BGS_discrete}
H(\bq|\bpi):=\sum\limits_{i=0}^{N-k} \frac{q_i}{\pi_i}\log\left(\frac{q_i}{\pi_i}\right)\pi_i
=
\sum\limits_{i=0}^{N-k} q_i\log\left(\frac{q_i}{\pi_i}\right).
\end{equation}

Note that we use here the definition of entropy with reverted sign with respect to the historical definition -- and also most common among physicists. 
Therefore, we expect its value to be non-increasing in time.

Let $\beta$ be the logarithmic mean
\begin{equation}
\label{eq:def_log_mean}
\beta(x,y):=
\left\{
\begin{array}{ll}
\frac{x-y}{\log x-\log y} & \mbox{if }x\neq y\\
x & \mbox{otherwise}
\end{array}
\right..
\end{equation}
J. Maas defined the following discrete optimal transport distance between probability densities:
\begin{definition}[Discrete Wasserstein distance \cite{Maas_JFA}]
\label{def:disc_Wasserstein}
Let $\bK$ be a stochastic, irreducible, and reversible transition kernel, and let $\bpi$ denote its unique stationary measure.
Given two probability densities $\bQ^0=\frac{\rd\bq^0}{\rd\bpi},\bQ^1=\frac{\rd\bq^1}{\rd\bpi}$ with respect to $\bpi$, the discrete squared Wasserstein distance is
\begin{equation}
\label{eq:def_action_discrete_Wasserstein}
\mathcal W_N^2(\bQ^0,\bQ^1):=
\inf \limits_{\bQ,\bpsi}
\left\{
\frac{1}{2}\int_0^1\sum \limits_{i,j}|\psi^\tau_i-\psi^\tau_j|^2K_{ij}\, \beta(Q^\tau_i,Q^\tau_j)\pi_i\,\rd \tau
\right\},
\end{equation}
where the infimum runs over all piecewise $\mathcal C^1$ curves of probability densities $\bQ:[0,1]\ni \tau\mapsto \bQ^\tau\in \mathbb R_+^N$ and all measurable functions  $\bpsi:[0,1]\ni \tau\mapsto \bpsi^\tau\in\R^N$ satisfying the discrete continuity equation with endpoints $\bQ^0,\bQ^1$
\begin{equation}
\label{eq:discrete_continuity}
\left\{
\begin{array}{l}
 \frac{\rd}{\rd\tau}Q^\tau_i + \sum\limits_j (\psi^\tau_j-\psi^\tau_i)K_{ij} \, \beta(Q^\tau_i,Q^\tau_j)=0\quad \,\text{a.e. }\tau\in[0,1], \, \forall i\ ,\\
 \bQ^\tau|_{\tau=0}=\bQ^0,\qquad \bQ^\tau|_{\tau=1}=\bQ^1.
\end{array}
\right.
\end{equation}
\end{definition}
This is a discrete counterpart to the celebrated Benamou-Brenier formula \cite{benamou2000computational} for the continuous Wasserstein distance \cite{Villani_Topics,santambrogio2015optimal}, and we emphasise the dependence of $\W_N$ on $N$,
since we will typically consider the large population limit $N\to\infty$ later on.

In what follows, we will slightly abuse notation  by noticing that one can canonically define the discrete Wasserstein distance between probabilities in terms of their densities $\W_N(\bq_0,\bq_1)\bydef\W_N\left(\frac{\rd \bq_0}{\rd\bpi},\frac{\rd \bq_1}{\rd\bpi}\right)$. In the sequel we will keep abusing the notations, and we shall simply speak of \emph{the} discrete Wasserstein distance and Riemannian structure 
	-- see e.g. \cite{erbar2016gradient} for similar issues on the subtle distinction between measures and densities.
Likewise, we will also write $H(\bq|\bpi)$ for $H\left(\bQ\circ\bpi|\bpi\right)$.

\begin{theorem}[Properties of the discrete Wasserstein distance \cite{Maas_JFA}]
\label{theo:prop_discrete_wasserstein}
With the same assumptions as before, denote by $\mathcal{P}_{\mathrm{dens}}$ the space of (finite) probability densities with respect to $\bpi$, and by $\mathcal P_{\mathrm{dens}}^*$ the subspace of everywhere strictly positive densities.
Then
\begin{enumerate}[(i)]
 \item 
 $\mathcal W_N$ defines a distance on $\mathcal P_{\mathrm{dens}}$;
 \item
 \label{item:discrete_Riemannian}
 The metric space $(\mathcal P_{\mathrm{dens}}^*,\mathcal W_N)$ is a Riemannian manifold;
 \item
 \label{item:grad_WN_coordinate}
 
 Given $f:\R^+\to \R$, the gradient of a functional $F(\bQ):=\sum_i f\left(Q_i\right)\pi_i$ with respect to the Riemannian structure in \eqref{item:discrete_Riemannian} reads, in local coordinates,
 \begin{equation}
 \label{eq:def_discrete_grad_WN}
 \left[\grad_{\mathcal W_N} F(\bQ)\right]_i=-\sum\limits_j K_{ij}\beta(Q_i,Q_j) \left(f'(Q_j)-f'(Q_i)\right).
 \end{equation}
 \item
 \label{item:discrete_grad_flow_Mass}
 The three equations of Lemma~\ref{lem:three_equivalent} are equivalent to the gradient flow
 \[
 \frac{\rd\bq}{\rd t}=-\grad_{\mathcal W_N} H(\bq|\bpi)
  \]
  of the relative BGS entropy.
\end{enumerate}
\end{theorem}
In \eqref{item:grad_WN_coordinate} and \eqref{item:discrete_grad_flow_Mass} the intrinsic Riemannian gradient $\grad_{\mathcal W_N} F(\bQ)$ of a function $F:\mathcal P_{\mathrm{dens}}(\mathcal X_N)\to\R$ is defined such that, along any differentiable curve $\tau\mapsto \bQ^\tau$ with speed $\dot \bQ^\tau\in T_{\bQ_t}\mathcal P_{\mathrm{dens}}$, the chain rule holds as
\[
\frac{\rd\ }{\rd\tau}F(\bQ^\tau)=\left\langle \grad_{\mathcal W_N} F(\bQ^\tau),\frac{\rd\bQ^\tau}{\rd\tau}\right\rangle_{\bQ^\tau}.
\]
The precise definition of the scalar product $\langle.,.\rangle_{\bQ}$ in the tangent plane $T_\bQ \mathcal P_{\mathrm{dens}}$ at a point $\bQ$ involves a certain weighted Onsager operator $L_\bQ\bydef-\mathbf{\operatorname{div}}_{\bQ}(\boldsymbol{\nabla}\cdot)$, defined in terms of suitable discrete divergence and gradient operators $\mathbf{\operatorname{div}}_{\bQ},\boldsymbol{\nabla}$.
This allows to formally rewrite \eqref{eq:def_discrete_grad_WN} in a more intrinsic fashion as
\begin{equation}
 \label{eq:formula_gradient_discrete_Wasserstein}
\grad_{\mathcal W_N}F(\bQ)=-\mathbf{\operatorname{div}}_{\bQ}(\boldsymbol{\nabla} f'(\bQ)).
\end{equation}
This will have a clear counterpart later on in the continuous world, see in particular \eqref{eq:formula_gradient_Wasserstein} and section~\ref{subsec:Wasserstein}.
This Onsager operator precisely gives the one-to-one correspondence between $\psi^\tau$ and $\frac{\rd\bQ^\tau}{\rd\tau}=L_{\bQ^\tau}\psi^\tau$ implicitly appearing in the continuity equation \eqref{eq:discrete_continuity}; see~\cite[section 3]{Maas_JFA}.
The fact that \eqref{item:discrete_grad_flow_Mass} holds is consequence of~\cite[Theorem 4.7]{Maas_JFA}.
This can be checked directly with \eqref{eq:def_discrete_grad_WN}:
with $F(\bQ)=H(\bQ\circ\bpi|\bpi)=\sum_i \{Q_i\log Q_i\}\pi_i=\sum_i\{Q_i\log Q_i-Q_i+1\}\pi_i$ we have $f'(Q_i)=\log (Q_i)$ in \eqref{eq:def_discrete_grad_WN}, hence
\begin{align*}
\left[\grad_{\mathcal W_N} H(\bQ\circ\bpi|\bpi)\right]_i&=-\sum\limits_j K_{ij}\beta(Q_i,Q_j) \left(f'(Q_j)-f'(Q_i)\right)
\\
&=-\sum\limits_j K_{ij}\frac{Q_j-Q_i}{\log Q_j-\log Q_i}\left(\log(Q_j)-\log(Q_i)\right)
 \\
&=-\left(\sum\limits_j K_{ij} Q_j\right) + \left(\sum\limits_j K_{ij}\right) Q_i\\
&=-(\bK \bQ)_i + Q_i=\left[(-\bK+\bI)\bQ\right]_i
\end{align*}
and $\dfrac{\rd\bQ}{\rd t}=-\grad_{\mathcal W_N} H(\bQ\circ\bpi|\bpi)$ reads indeed $\dfrac{\rd\bQ }{\rd t}= \left(\bK-\bI\right)\bQ$ as in Lemma~\ref{lem:three_equivalent}.

As pointed out in \cite{Maas_JFA}, the restriction to \emph{positive} densities in \eqref{item:discrete_Riemannian} is not an issue:
since the kernel $\bK$ is irreducible any solution of the Heat Equation $\frac{\rd\bQ}{\rd t}=(\bK-\bI)\bQ$ becomes instantaneously positive, $Q_i(t)>0$ for all $i$ and $t>0$.

We would like to stress that our main interest lies in \eqref{item:discrete_grad_flow_Mass}:
although the original evolution $\frac{\rd\bp}{\rd t}=(\bM-\bI)\bp$ is not truly speaking a gradient flow (due to absorbing states causing reducibility and irreversibility) one can in fact change the relevant variables so that the new effective (lower dimensional) $Q$-process kernel $\bK$ becomes irreducible and reversible, and obtain a variational structure via discrete mass transport.
Summarising the previous discussions, we have established in this section:
\begin{theorem}[gradient flow structure for reducible irreversible kernels]
	\label{thm:vs_rik}
Let $\bM$ be admissible in the sense of Definition~\ref{def:ac}, let $\bK$ be the associated $Q$-process with stationary measure $\bpi$, and let $\bq(t)$ be defined by~\eqref{eq:rescaling}.
Then the evolution of  $\bp$ is ``variational''  in the sense that $\bq$ is driven by the gradient flow
\begin{equation}
\label{eq:discrete_grad_flow_Q}
 \frac{\rd\bq}{\rd t}=-\grad_{\mathcal W_N} H(\bq|\bpi).
\end{equation}
and that the complete dynamics of $\bp$  can be uniquely recovered from $\bq$. 
\end{theorem}

 \begin{remark}\label{rmk:equiv_pq}
 The fact that the evolution of $\bq$ fully determines the dynamics of $\bp$ can be seen by undoing the change of variables from $\widetilde{\bp}$ to $\bq$ and then using the conservation laws to recover $\bp$, see Lemma~\ref{lem:expl_sg} and Equation~\eqref{eq:rescaling}.
 \end{remark}

The above identification of the gradient flow structure only involved the BGS entropy $H(\bq|\bpi)$ as a driving functional.
One can also define more general relative $\phi$-entropies of the form
\begin{equation}
\label{eq:def_entropy_tsallis_discrete}
G_\phi(\bq|\bpi)\bydef\sum\limits_{i=0}^{N-k} \phi\left(\frac{q_i}{\pi_i}\right)\pi_i
\end{equation}
for a given convex function $\phi:\R^+\to\R$, still computed relatively to the reference measure $\bpi$.
This covers the so-called Tsallis entropies, a generalisation of the Boltzmann-Gibbs entropy for non-additive systems with growing importance in biology (see~\cite{Xuan_etal} and references therein); see also~\cite{Karev} to its application in a model of prebiotic evolution akin to the replicator equation.
Tsallis entropies will be further discussed in Section~\ref{sec:DTDMC}.%

It turns out that the $G_\phi$ functional is a Lyapunov functional for the $\bq$-evolution (which is rather driven by the $H$ entropy.)

\begin{lemma}\label{lem:entropy_CD}
Let $\bM\in\mathcal A_{N+1,k}$ be an admissible and micro-reversible transition kernel, let $G_\phi$ be as in~\eqref{eq:def_entropy_tsallis_discrete} for a given differentiable convex function $\phi:\R^+\to\R$, and consider a solution $\bQ(t)$ of the previous gradient flow (Lemma~\ref{lem:three_equivalent} and Theorem~\ref{theo:prop_discrete_wasserstein}).
Then
\begin{equation}
\label{eq:cross_dissipation}
  \frac{\rd\ }{\rd t}G_\phi(\bQ|\bpi)
  =-\frac{1}{2}\sum_{i,j=0}^{N-k}K_{ij}\beta(Q_i,Q_j) \left(\phi'\left(Q_j\right)-\phi'\left(Q_i\right)\right)(\log Q_j-\log Q_i)\pi_i\leq 0.
\end{equation}
Moreover, if $\phi$ is locally strongly convex in the sense that $\phi''(x)\geq c_M>0$ in any bounded interval $x\in[0,M]$, then there exists $C>0$ depending only on $\phi,\bpi$ (but not explicitly on $\bK$ or on the solution $\bQ$) such that there holds the improved dissipation estimate
\begin{equation}
 \label{eq:cross_dissipation_H}
\frac{\rd\ }{\rd t}G_\phi(\bQ|\bpi) \leq -C\left|\frac{\rd\bQ}{\rd t}\right|_{\bpi}^2\leq 0
\end{equation}
Furthermore, if $\phi(x)\ge x-1$, then $G_\phi(\bQ|\bpi)\ge0$.
\end{lemma}
We stress that the right-hand side in \eqref{eq:cross_dissipation} is nothing but the expression in local coordinates of the Riemannian chain rule
\begin{align*}
 \frac{\rd\ }{\rd t}G_\phi(\bQ|\bpi)
 &=\langle\grad_{\mathcal W_N} G_{\phi}(\bQ|\bpi),\frac{\rd}{\rd t}\bQ\rangle
 \\
&=-\langle\grad_{\mathcal W_N} G_{\phi}(\bQ|\bpi),\grad_{\mathcal W_N} H(\bQ|\bpi)\rangle.
\end{align*}
That this quantity is indeed nonnegative simply follows from our assumption that $\phi'$ is nondecreasing, the summand in \eqref{eq:cross_dissipation} being nonnegative for all $i,j$.
In \eqref{eq:cross_dissipation_H} we stress that $|.|^2_{\bpi}$ is the squared norm induced by the weighted scalar product $\langle a,b\rangle_{\bpi}=\sum a_ib_i\,{\pi}_i$ introduced earlier, which should not be confused with the Riemannian scalar product $\langle.,.\rangle_{\bQ}$ at a point $\bQ$.
Note that the right-hand side vanishes if and only if $\bQ$ is a stationary point, $\frac{\rd\bQ}{\rd t}=0\Leftrightarrow \bQ=\bK\bQ$.
\begin{proof}
Let $\bK$ be given by \eqref{eq:def_K}, with $\frac{\rd\bQ}{\rd t}=(\bK-\bI)\bQ$ and $Q_i(t)=\frac{q_i(t)}{\pi_i}$.
We compute
 \begin{align*}
  \frac{\rd\ }{\rd t}G_\phi(\bQ|\bpi)
  &=
  \frac{\rd\ }{\rd t}\left(\sum_{i=0}^{N-k} \phi\left(Q_i\right)\pi_i\right)\\
  &=
  \sum_{i=0}^{N-k}\phi'\left(Q_i\right)\frac{\rd Q_i}{\rd t}\pi_i
  =\sum_{i=0}^{N-k}\phi'\left(Q_i\right)\left[\sum_{j=0}^{N-k}K_{ij}Q_j-Q_i\right]\pi_i
  \\
  &=
  \sum_{i=0}^{N-k}\phi'\left(Q_i\right)\left[\sum_{j=0}^{N-k}K_{ij}(Q_j-Q_i)\right]\pi_i,
  \end{align*}
  where the last equality simply comes from the stochasticity $\sum_j K_{ij}=1$.
Leveraging the microreversibility $K_{ij}\pi_i=K_{ji}\pi_j$, a straightforward summation by parts leads to
\begin{align*}
  \frac{\rd\ }{\rd t}G_\phi(\bQ|\bpi)
  &=
  -\frac{1}{2}\sum_{i,j=0}^{N-k}K_{ij}\left(\phi'\left(Q_j\right)-\phi'\left(Q_i\right)\right)(Q_j-Q_i)\pi_i
  \\
  &=-\frac{1}{2}\sum_{i,j=0}^{N-k}K_{ij}\frac{Q_j-Q_i}{\log Q_j-\log Q_i}  \left(\phi'\left(Q_j\right)-\phi'\left(Q_i\right)\right)(\log Q_j-\log Q_i)\pi_i\ .
\end{align*}
Recalling the definition of the logarithmic mean, Equation~\eqref{eq:def_log_mean}, this proves Equation~\eqref{eq:cross_dissipation}.

Assume now that $\phi$ satisfies the additional strong local convexity as in our statement.
Just as before we can write
\begin{equation}
\label{eq:dGphi}
\frac{\rd\ }{\rd t}G_\phi(\bQ|\bpi)
=\sum_{i=0}^{N-k}\phi'\left(Q_i\right)
\left[\sum_{j=0}^{N-k}K_{ij}Q_j-Q_i\right]\pi_i
\end{equation}
Observe that due to $Q_i=\frac{q_i}{\pi_i}\leq \frac{\sum_j q_j}{\pi_i}=\frac{1}{\pi_i}$ we have $0\leq Q_i\leq M:=\frac{1}{\min_j \pi_j}<+\infty$.
Recalling that $\bK$ is row-stochastic we see that the convex combination $\sum _{j=0}^{N-k}K_{ij}Q_j\leq M$ as well, for any fixed $i$.
Exploiting our strong convexity assumption on $\phi$, a straightforward application of Taylor's theorem gives that
\[
\phi'\left(Q_i\right)
\left[\sum_{j=0}^{N-k}K_{ij}Q_j-Q_i\right]
\leq
\phi\left(\sum_{j=0}^{N-k}K_{ij}Q_j\right)-\phi(Q_i)
- \frac{c_M}2\left|
\sum_{j=0}^{N-k}K_{ij}Q_j-Q_i
\right|^2
\]
where $c_M>0$ is a lower bound for $\phi''(x)$ in the interval $[0,M]$.
Substituting in \eqref{eq:dGphi} gives then, by convexity of $\phi$ and reversibility $K_{ij}\pi_i=K_{ji}\pi_j$,
\begin{align*}
\frac{\rd\ }{\rd t}G_\phi(\bQ|\bpi)
&\leq
\sum_{i=0}^{N-k} 
\left\{\phi\left(\sum_{j=0}^{N-k}K_{ij}Q_j\right)-\phi(Q_i)
-C\left|
\sum_{j=0}^{N-k}K_{ij}Q_j-Q_i
\right|^2
\right\}\pi_i
\\
&\leq 
\sum_{i=0}^{N-k} 
\left\{\sum_{j=0}^{N-k}K_{ij}\phi\left(Q_j\right)-\phi(Q_i)
-C\left|
\sum_{j=0}^{N-k}K_{ij}Q_j-Q_i
\right|^2
\right\}\pi_i
\\
&=
\sum_{j=0}^{N-k}\sum_{i=0}^{N-k} K_{ji}\phi(Q_j)\pi_j-\sum\limits_{i=0}^{N-k}\phi(Q_i)\pi_i -C \sum_{i=0}^{N-k}\left|\frac{\rd Q_i}{\rd t}\right|^2\pi_i
\\
&=
\sum_{j=0}^{N-k}\underbrace{\left(\sum_{i=0}^{N-k} K_{ji}\right)}_{=1}\phi(Q_j)\pi_j-\sum\limits_{i=0}^{N-k}\phi(Q_i)\pi_i -C \sum_{i=0}^{N-k}\left|\frac{\rd Q_i}{\rd t}\right|^2\pi_i
\\
&=
\sum\limits_{j=0}^{N-k}\phi(Q_j)\pi_j-
\sum\limits_{i=0}^{N-k}\phi(Q_i)\pi_i -C \sum_{i=0}^{N-k}\left|\frac{\rd Q_i}{\rd t}\right|^2\pi_i
\\
&=-C \sum_{i=0}^{N-k}\left|\frac{\rd Q_i}{\rd t}\right|^2\pi_i
= -C\left|\frac{\rd\bQ}{\rd t}\right|_{\bpi}^2
\end{align*}
as desired.

Finally, the fact that $G_\phi(\bQ|\bpi)\geq 0$ immediately follows from
\[
G_\phi(\bQ|\bpi)
  =
  \sum_{i=0}^{N-k} \phi\left(Q_i\right)\pi_i
  \geq \sum_{i=0}^{N-k} \left(Q_i-1\right)\pi_i=\sum_{i=0}^{N-k} Q_i\pi_i-\sum_{i=0}^{N-k} \pi_i=0,
\]
where we used $\sum_i Q_i\pi_i=\sum_iq_i=\sum_i\pi_i=1$. \qed
\end{proof}

We will state later on in Section~\ref{sec:DTDMC} a discrete-in time monotonicity of the entropy, Proposition~\ref{lem:decrasing}.
This discrete-time framework will be based on a ``Euclidean'' discretization (linear 1st order difference quotient), and we believe that in this context a discrete equivalent of \eqref{eq:cross_dissipation_H} can be established.
However, \eqref{eq:cross_dissipation} strongly leverages the Riemannian structure from discrete optimal transport, and seems therefore incompatible with the Euclidean time discretization.
An alternative and natural discretization in time would rather be the JKO scheme, more compatible with the discrete Wasserstein and gradient flow structure (see Section~\ref{sec:gfcmp} for a short discussion). For such JKO discretization, it should be possible to write down a discrete version of \eqref{eq:cross_dissipation}. However, even if this turns out to be possible, the relationship of this new discrete process with the embedded chain  is not clear, and we will not pursue this direction here.

We also point out that $G_\phi=H$ is admissible in Lemma~\ref{lem:entropy_CD}, in which case one is actually computing the dissipation of $H$ along its own gradient flow and the right-hand side of \eqref{eq:cross_dissipation} is a discrete version of the Fisher information $F(\rho|\pi)=\int \left|\nabla \log\left(\frac{\rho}{\pi}\right)\right|^2\rho$.

\subsection{The continuous in time Moran process}
\label{ssec:moran_proc}

In this subsection, we study the Moran process with the techniques developed above. 
As already explained, the only two absorbing states are labelled here $i=0$ and $i=N$.
We start with explicit results for the so-called \emph{neutral} evolution, in which both types $\A$ and $\B$ have the same reproductive viability, namely, $\tsp_i^{(\mathsf{n})}=i/N$ in the transition matrix given by Equation~\eqref{def:Moran_TM}. We indicate the neutral evolution in this subsection by the superscript $\mathsf{n}$.
We will obtain an explicit formula for the Wasserstein distance between adjacent sites in birth and death processes, and finish with some comments on more general evolutionary processes including the Wright-Fisher process.

For the neutral Moran process, a simple calculation shows that $z_i^{(\mathsf{n})}=\frac{1}{N-1}$, $w_i^{(\mathsf{n})}=\frac{6i(N-i)}{N(N+1)}$ and $\lambda^{(\mathsf{n})}=-\frac{2}{N^2}$.
Introducing  the auxiliary notation $x_i=i/N$ and considering the  BGS entropy given by Equation~\eqref{eq:def_entropy_BGS_discrete} as a function of $\bp$ (rather than $\bq$) through the change of variables \eqref{eq:rescaling}, we arrive at the entropy for the  neutral evolution:
\begin{equation}\label{eq:Moran_neutral_Entropy}
H^{(\mathsf{n})}(\bp|\bpi^{(\mathsf{n})})=\sum_{i=1}^{N-1}\frac{x_i(1-x_i)p_i}{\sum_jx_j(1-x_j)p_j}\log\left(\frac{N^2-1}{6N}\frac{p_i}{\sum_jx_j(1-x_j)p_j}\right)\ .
\end{equation}

For non-neutral Moran process, one should not expect analytical formulas for  $\lambda$, $\bw$ and $\bz$. However, as a consequence of Lemma~\ref{lem:tridiagonal_microreversible} and Equation~\eqref{eq:generalized_detailed_balance}, we find that
\[
 \frac{w_{i+1}}{z_{i+1}}=\frac{M_{i,i+1}}{M_{i+1,i}}\times\frac{w_i}{z_i}=\frac{(i+1)(1-\tsp_{i+1})}{(N-i)\tsp_i}\frac{w_i}{z_i}
=\frac{(i+1)!(N-i-1)!\prod_{j=2}^{i+1}(1-\tsp_j)}{(N-1)!\prod_{j=1}^i\tsp_j}\frac{w_1}{z_1}\ .
\]
Therefore
\begin{equation}\label{eq:ratio_ev}
\frac{w_i}{z_i}=CN\binom{N}{i}^{-1}\frac{\prod_{j=1}^i(1-\tsp_j)}{\prod_{j=1}^{i-1}\tsp_j}\ ,\quad i=1,\ldots,N-1.
\end{equation}
where $C$ is a certain normalisation constant.
This equation will be used in Subsection~\ref{ssec:kimura_eq} to identify the correct macroscopic limit of the Moran process, i.e., the correct model when $N\to\infty$.

For general Moran processes, we are able to compute explicitly the Wasserstein distance between adjacent sites. Namely,
\begin{lemma}
\label{lem:1st_order_expansion_WN}
Let $\bK$ be a tridiagonal, irreducible, and reversible transition matrix, let $\beta(x,y)$ be the logarithmic mean defined in~\eqref{eq:def_log_mean}, and let us write $\be_i$ for the discrete probability vector concentrated on the $i$-th state.
For adjacent sites $j=i+1$ we have
\begin{equation}
\label{eq:1st_order_expansion_WN}
 \W_N(\be_i,\be_{i+1})=\int_0^1\frac{\rd r}{\sqrt{\beta\big(K_{\small i+1,i}r,K_{\small i,i+1}(1-r)\big)}}\ .
\end{equation}
\end{lemma}
\begin{proof}
Since $\bK$ is tridiagonal, the minimising curve in Definition~\ref{def:disc_Wasserstein} should only involve the two neighbouring sites, i.e. the density $\bQ^\tau=r(\tau)\frac{\be_i}{\pi_i}+(1-r(\tau))\frac{\be_{i+1}}{\pi_{i+1}}$ for some function $r$ with $r(0)=1$ and $r(1)=0$ to be determined.
From~\eqref{eq:discrete_continuity}, we conclude that
\[
 \psi_{i+1}^\tau-\psi_i^\tau=-\frac{\dot Q^\tau_i}{K_{i,i+1}\beta(Q^\tau_i,Q^\tau_{i+1})}\ ,
\]
where $\dot{}=\frac{\rd\ }{\rd\tau}$.
Plugging this into the action~\eqref{eq:def_action_discrete_Wasserstein} and leveraging the reversibility of $\bK$, we find a functional of $r$ only. Namely
\begin{align*}
&\frac{1}{2}\int_0^1\sum \limits_{j,k}|\psi^\tau_j-\psi^\tau_k|^2K_{jk}\, \beta(Q^\tau_j,Q^\tau_k)\pi_j\,\rd \tau
\\
&\quad=
\frac{1}{2}\int_0^1\left|\psi^\tau_i-\psi^\tau_{i+1}\right|^2K_{i,i+1}\beta(Q_i^\tau,Q_{i+1}^\tau)\pi_i\rd\tau
\\
&\qquad+
\frac{1}{2}\int_0^1\left|\psi^\tau_{i+1}-\psi^\tau_{i}\right|^2K_{i+1,i}\beta(Q_{i+1}^\tau,Q_{i}^\tau)\pi_{i+1}\rd\tau
\\
&\quad=
\int_0^1\frac{\left(\dot Q^\tau_i\right)^2\pi_i\rd\tau}{K_{i,i+1}\beta(Q^\tau_i,Q^\tau_{i+1})}
=
\int_0^1\frac{\left(\frac{r'(\tau)}{\pi_i}\right)^2\pi_i\rd\tau}{K_{i,i+1}\beta\left(\frac{r(\tau)}{\pi_i},\frac{1-r(\tau)}{\pi_{i+1}}\right)}
\\
&\quad=\int_0^1\frac{r'(\tau)^2\rd\tau}{\beta(K_{i,i+1}r(\tau),K_{i+1,i}(1-r(\tau))}\ .
\end{align*}
Thus $\W_N^2(\be_i,\be_{i+1})$ is the minimal value of this functional over all $r\in\mathcal C^1([0,1])$, such that $r(0)=1$, $r(1)=0$, the Euler-Lagrange Equation implies that
\[
 \frac{\rd\ }{\rd\tau}\left(\frac{r'(\tau)^2}{F(r(\tau))}\right)=0\ ,
\]
i.e., $r'(\tau)=c\sqrt{F(r(\tau))}$, where $F(r)=\beta(K_{i,i+1}r,K_{i+1,i}(1-r))$ and 
\[
 c=\int_0^1\frac{r'(\tau)\rd\tau}{\sqrt{F(r(\tau))}}=\int_0^1\frac{\rd r}{\sqrt{\beta(K_{i,i+1}r,K_{i+1,i}(1-r))}}\ .
\]
The result follows immediately. \qed
\end{proof}

Note that the fact that the curve linking two neighbouring states is a linear combination of $\be_i$ and $\be_{i+1}$ only holds here because the matrix is tridiagonal.
Otherwise the nonlocal effects start building up and many other states play a role too, and the problem becomes more complicated than the optimisation over scalar functions $r(\tau)$ above.

We finish with one more comment about general processes in the Kimura class:
\begin{remark}
	For general $\bM\in\mathcal{A}_{N+1,2}$, the micro-reversibility does not necessarily hold. In particular, for the Wright-Fisher process this is not true even in the neutral case. However, 
for sufficiently large $n$ 
\[
 \bM^n\approx\left(
 \begin{matrix}
  1&(\bone-\bF)^\dagger&0\\
  \bzero &\lambda^n\widetilde{\bz}\otimes\widetilde{\bw}&\bzero\\
  0&\bF^\dagger&1
 \end{matrix}\right)\ ,
\]
where $\bF$ is the unique solution of $\bF^\dagger=\bF^\dagger\bM$ with $F_0=0$, $F_N=1$. 
Therefore ${(\bM^n)_{ij}\approx\lambda^n z_iw_j}$ for $i,j=1,\dots,N-1$, i.e., $w_iM_{ij}^n z_j\approx w_jM_{ji}^n z_i$. In a loose sense, we would say that any process in the Kimura class is asymptotically micro-reversible. However, we will not explore these ideas in the current work.
\end{remark}

\section{Interlude: discrete in time, discrete Markov chains}
\label{sec:DTDMC}

The notion of relative entropy, and even the entropy itself, is usually defined for \emph{irreducible} Markov processes only, see e.g. \cite{Sober_2011} for BGS entropy, and \cite{Furuichi_etal_2004} for Tsallis entropies.
Due to the existence of absorbing states in the transition kernel, this notion of course does not apply here to our admissible matrices.

In this section, we study discrete in time Markov processes.
Discrete in time models cannot be recast in the gradient flow formalism, however we opt to include this discussion in the present work because, as for the Moran process, these are frequently used in population genetics.
Furthermore, the results of the previous section can shed light on the precise notion of entropy for processes with two or more absorbing states.

We shall see that the right notion of entropy only depends on the transient states, and we will therefore speak of \emph{substochastic entropies}. 
We would like to point out that this issue is a priori non trivial: when $\bM$ has absorbing states, usual  entropy functionals often lead to quantities that increase along the evolution. (We recall that with our sign conventions entropies should rather be nonincreasing in time.)

We first show in Subsection~\ref{ssec:embedded_chain} how to make discrete in time and continuous in time models compatible, given a certain transition matrix $\bM$ for the continuous in time model.
In the sequel, we show that the entropy functional is independent of the time step and show how a generalisation of the BGS and Tsallis entropies for substochastic models naturally arises from our analysis.

\subsection{The embedded chain}
\label{ssec:embedded_chain}

Every stochastic matrix $\bM$  can always  be seen as a discrete Markov chain in either  discrete or continuous time. 
Indeed, $\bM$ can be seen as a transition matrix and the evolution of $\bp$ is given by Equation~\eqref{DTDMC_ee};  alternatively,  it can be seen as the so-called embedded chain, whose evolution is given by Equation~\eqref{CTDMC_ee}, see~\cite{norris1998markov}.
Let $h>0$ be a small time step and define the kernel 
\begin{equation}\label{eq:def_Mh}
\bM^{(h)}\bydef\bI+h\left(\bM-\bI\right)\ .
 \end{equation}
Let also $\bp^{(h)}(t)$ be the (piecewise constant) time-interpolation defined by the recursion
\begin{equation}\label{eq:discrete_evolution}
\bp^{(h)}(t+h)=\bM^{(h)}\bp^{(h)}(t) ,
\end{equation}
with $\bp^{(h)}(0)=\bp^\ini$.
From the fact that
\[
\bp^{(h)}(t+h)=\bM^{(h)}\bp^{(h)}(t)
\qquad \Leftrightarrow\qquad
\frac{\bp^{(h)}(t+h)-\bp^{(h)}(t)}{h}=(\bM-\bI)\bp^{(h)}(t)\ ,
\]
we see that~\eqref{eq:discrete_evolution} is of course the explicit Euler discretization of~\eqref{CTDMC_ee}.
Thus, the following convergence is not surprising:
\begin{lemma}\label{lem:disc_to_cont_time}
In the limit of small time steps we have
	\[
	\lim_{h\to 0}\bp^{(h)}(t)=\bp(t)\ ,
	\]
	where $\bp(t)$ is the solution of Equation~\eqref{CTDMC_ee} with $\bp(0)=\bp^\ini$ and the convergence is locally uniform in time.
\end{lemma}
\begin{proof}
	The convergence of the Euler scheme is a standard result, cf.~\cite{SB2002}. 
	In our particular context of linear ODEs, a simple proof consists in writing
	$n=\left\lfloor t/h\right\rfloor$, and noticing that
	$\lim_{h\to 0}\left(\bM^{(h)}\right)^n=\exp\left(t\left(\bM-\bI\right)\right)$. \qed
\end{proof}

\subsection{Dynamics and generalised entropies}
In this section we only consider admissible transition kernels, 
$\bM\in\mathcal A_{N+1,k}$.

Note that the entropy $H(\bq|\bpi)$, given by~\eqref{eq:def_entropy_BGS_discrete} as a function of the $(N+1)-k$ dimensional $Q$-process $\bq$, depends on $\bM$ only through the characteristic eigenvectors $\widetilde\bw,\widetilde\bz$ - since by definition the reference measure $\bpi=\widetilde \bw \circ \widetilde \bz$.
With this fact in mind, it is clear that changing the transition kernel from $\bM$ to $\bM^{(h)}$, for any $h>0$ will not change the entropy. More precisely,

\begin{lemma}\label{lem:caracteritic_triple_h}
Let $(\mu,\widetilde{\bw},\widetilde{\bz})$ be the characteristic triple of $\bM$ in the sense of Definition~\ref{def:characteristic_triple}, and let $\bM^{(h)}$ be the discrete-in-time transition kernel defined in~\eqref{eq:def_Mh}.
 Then $\bM^{(h)}$ has characteristic triple $(\mu_h,\widetilde{\bw},\widetilde{\bz})$, with 
 \[
 \mu_h\bydef1-h(1-\mu)\ .
\]
In particular $\bM$ and $\bM^{(h)}$ share their characteristic eigenvectors.
Moreover, $\bM^{(h)}$ is micro-reversible if and only if $\bM$ is.
Finally, if  $\bM\in\calA_{N+1,2}$ is a Kimura matrix with fixation probability $\bF$, then so is $\bM^{(h)}$.
\end{lemma}
\begin{proof}
 This easily follows from the expression~\eqref{eq:def_Mh} for $\bM^{(h)}$. \qed
\end{proof}
For notational convenience, we extend the previous transient reference measure $\bpi=(\pi_i)_{i=0,\dots, N-k}$  by $k$ zeros to form the corresponding full $(N+1)$-dimensional probability measure
\[
\overline\bpi:=(\bpi,\mathbf{0}).
\]
It is also natural to extend definition~\eqref{eq:def_entropy_tsallis_discrete} to the full dynamics, i.e., by abuse of notation, using the variable $\bp$ instead of $\bq$, which are equivalent in view of Remark~\ref{rmk:equiv_pq}, we write 
\begin{equation}\label{eq:def_entropy_tsallis_p}
 G_\phi(\bp|\overline\bpi)\bydef
 \begin{cases}
 \sum\limits_{i=0}^{N-k}\phi\left(\frac{\widetilde p_i}{\widetilde z_i\langle\widetilde \bw,\widetilde \bp\rangle}\right)\widetilde w_i\widetilde z_i 
 &\text{if}\quad\langle\widetilde\bw,\widetilde\bp\rangle\ne 0,
 \\
 +\infty 
 &\text{otherwise}.
 \end{cases}
\end{equation}
This definition is completely independent of the time step $h$.

\begin{remark}
The entropy is defined on $\bq$, and therefore on the transient states $\widetilde\bp$, and not on the whole probability $\bp$.
In terms of measure theoretic considerations, our definition can be summarised as follows: performing a Lebesgue decomposition with respect to the reference measure $\overline\bpi$, i.e.
$\bp=\widetilde\bp +\bp^{\mathrm{s}}$ with $\widetilde \bp\ll\overline\bpi$ and $\bp^{\mathrm{s}}\perp \overline\bpi$,
our entropy is finite if and only if $\widetilde \bp\neq 0$.
This departs from the usual definition of relative entropies, where one usually sets the entropy to $+\infty$ if $\bp$ is not absolutely continuous w.r.t. the reference measure $\overline\bpi$ (i.e. if the singular part $\bp^{\mathrm{s}}\neq 0$).
Here we have a whole freedom between those two scenarios, and our entropy takes finite values even for $\bp$'s that are partially absorbed with both $\widetilde\bp\neq 0$ and $\bp^{\mathrm{s}}\neq 0$.
Note also that our reference measure $\overline\bpi$ does not have full support.
\end{remark}

For the reader aiming to apply the previous results in the traditional (i.e., discrete time) Moran process, we state the entropy inequality  in the natural variables.
Namely, taking as primary variables the probability distribution $\bp$ and considering generalised entropies, we state the following result.
\begin{proposition}[Discrete monotonicity of generalised entropies]\label{lem:decrasing}
Let $\bM\in\mathcal A_{N+1,k}$ be an admissible and micro-reversible transition kernel, let $\phi:\R^+\to\R$ be a convex function, and let $G_\phi(\cdot|\overline\bpi)$ be the associated generalised entropy from Definition~\ref{eq:def_entropy_tsallis_discrete}. 
Then the one-step monotonicity
\[
G_\phi(\bM^{(h)}\bp|\overline\bpi)\le G_\phi(\bp|\overline\bpi)
\]
 holds for any sufficiently small time step $h$ (depending only on $\bM$ but not on $\bp$).
 Furthermore, if $\phi(x)\ge x-1$, then $G_\phi(\bp|\overline\bpi)\ge 0$ for any probability vector $\bp$.
\end{proposition}
\begin{proof}
Let $(\mu,\widetilde\bw,\widetilde\bz)$ be the characteristic triple of $\bM$ and $h\in(0,1/(1-\mu))$.
If $G_\phi(\bp|\overline\bpi)=+\infty$ there is nothing to prove, hence we only consider $\langle\widetilde\bw,\widetilde\bp\rangle >0$.
Let us first prove that $\bM^{(h)}\bp$ has finite entropy, i.e. $\langle\widetilde \bw, \widetilde{\bM^{(h)}\bp}\rangle>0$ as well.
To this end we first extend $\widetilde\bw$ to $\overline\bw$ by zeros.
Then, since $\langle \widetilde\bw,\widetilde\bp\rangle=\langle\overline\bw,\bp\rangle=\overline\bw^\dagger\bp$ for any $\bp$, we get
\begin{equation}
\label{eq:discrete_exponential_decay_transient}
 \langle\widetilde \bw, \widetilde{\bM^{(h)}\bp}\rangle
 = \overline\bw^\dagger\bM^{(h)}\bp
 = \mu_h\overline\bw^\dagger \bp
 = \mu_h\langle\widetilde\bw,\widetilde\bp\rangle >0
\end{equation}
if the time step $h\in(0,1/(1-\mu))$ is small enough to guarantee $\mu_h=1-h(1-\mu)>0$.
This proves that $\bM^{(h)}\bp$ cannot be completely absorbed if $\bp$ is not, hence $G_\phi(\bM^{(h)}\bp|\overline\bpi)<\infty$.

The remainder of the proof follows from Lemma~\ref{lem:caracteritic_triple_h}, identity~\eqref{eq:discrete_exponential_decay_transient}, and the convexity of $\phi$:
\begin{align*}
 G_\phi(\bM^{(h)}\bp|\overline\bpi)&=\sum_{i=0}^{N-k}\phi\left(\frac{(\widetilde{\bM^{(h)}\bp})_i}{\widetilde{z_i}\langle\widetilde{\bw},\widetilde{\bM^{(h)}\bp}\rangle}\right)\widetilde{w_i}\widetilde{z_i}\\
 &=\sum_{i=0}^{N-k}\phi\left(\frac{\sum_{j=0}^{N}M^{(h)}_{ij}p_j}{\widetilde{z_i}\mu^{(h)}\langle\widetilde{\bw},\widetilde{\bp}\rangle}\right)\widetilde{w_i}\widetilde{z_i}\\
 &=\sum_{i=0}^{N-k}\phi\left(\sum_{j=0}^{N}\frac{\widetilde{w}_jM^{(h)}_{ji}}{\mu^{(h)}\, \widetilde{w_i}}\frac{p_j}{\widetilde{z}_j\langle\widetilde{\bw},\widetilde{\bp}\rangle}\right)\widetilde{w_i}\widetilde{z_i}\\
&\le\sum_{i=0}^{N-k}\sum_{j=0}^{N}\frac{\widetilde{w}_jM^{(h)}_{ji}}{\mu^{(h)}\, \widetilde{w}_i}\phi\left(\frac{p_j}{\widetilde{z}_j\langle\widetilde{\bw},\widetilde{\bp}\rangle}\right)\widetilde{w_i}\widetilde{z_i}\\
 &=\sum_{j=0}^{N}\phi\left(\frac{p_j}{\widetilde{z}_j\langle\widetilde{\bw},\widetilde{\bp}\rangle}\right)\widetilde{w}_j\widetilde{z}_j
 = G_\phi(\bp|\overline{\bpi})\ .
\end{align*}
\qed
\end{proof}

\begin{remark}
Choosing $\phi(x)=x\log x\geq x-1$ we recover the previous expression~\eqref{eq:def_entropy_BGS_discrete}, providing a generalisation of the relative Boltzmann-Gibbs-Shannon entropy for reducible processes.
On the other hand, choosing $\phi(x)=\frac{x^m-x}{m-1}\geq x-1$ for exponents $m>1$ also provides a generalisation of the relative Tsallis entropies for reducible kernels, i.e.
\[
G_{\text{Tsallis}}(\bp|\overline\bpi)
=\sum_{i=0}^{N-k} \frac{\left(\frac{p_i}{\widetilde z_i\langle\widetilde \bw,\widetilde\bp\rangle}\right)^m-\frac{p_i}{\widetilde z_i\langle\widetilde{\bw},\widetilde{\bp}\rangle}}{m-1}\widetilde w_i\widetilde z_i
=\frac{1}{m-1}\left[\sum_{i=0}^{N-k}\frac{\widetilde w_ip_i^m}{\widetilde z_i^{m-1}\langle\widetilde \bw,\widetilde \bp\rangle^m}-1\right]\ .
\]

When $\bM$ is irreducible we have $k=0$, $\widetilde\bw=\bone$, and the reference measure is $\overline\bpi=\widetilde\bz$: in this case our generalised Tsallis entropies coincides with that in \cite{Furuichi_etal_2004}.
Let us also point out that, at the continuous level $x\in[0,1]$, both the relative Kullback-Leibler divergence $\mathcal H(\rho|\pi)=\int\rho\log(\rho/\pi)\rd x=\int \rho\log\rho\,\rd x+\int \rho V\,\rd x$ (with Gibbs distribution $\pi(x)=e^{-V(x)}/Z$) and the Tsallis entropies $\mathcal E_m(\rho)=\int \frac{\rho^m-\rho}{m-1}\rd x$ play a particular role in continuous optimal transport, since the Fokker-Planck Equation $\partial_t\rho=\Delta\rho+\dive(\rho\nabla V)$ and the Porous Medium Equation $\partial_t\rho=\Delta\rho^m$ can be viewed as their respective Wasserstein gradient flows, see \cite{Jordan_Kindeleherer_Otto,otto2001geometry,santambrogio2015optimal,Villani_Topics} and the next section.
\end{remark}
 
 \begin{remark}
  Results from Lemma~\ref{lem:entropy_CD} apply for Tsallis entropies in the range $m\in(1,2]$.
  This particular range seems to be of relevance to applications in ecology, as one possible interpretation of $m$ is as an interpolation parameter between two well established diversity measures for populations: the Shannon-Wiener index (the BGS entropy, in our notation) in the limit $m\to 1$ and the Simpson index, defined as the limit $m\to2$ of the Tsallis entropies; see~\cite{Xuan_etal}.
 \end{remark}
 
\section{Continuous time, continuous Markov chains}\label{sec:CTCMC}
In this section we take interest in the gradient flow formulation of the Kimura Equation~\eqref{eq:KimuraPDE}. The relation between continuous and discrete models will be investigated in Section~\ref{sec:gfcmp}, but the reader may want to keep in mind the following picture:
The Kimura Equation is the continuous counterpart of the Moran process, where $x\in[0,1]$ plays the role of $i/N$, i.e., the fraction of individuals of a given type in the population.
Although motivated by this large population limit, we focus in this section on a self-contained presentation of a slightly more general setting than the specific Kimura PDE \eqref{eq:KimuraPDE}.

We start by defining this generalised Kimura Equation in Subsection~\ref{subsec:background}.
We continue with the exploration of the two ingredients of the gradient flow: an entropy in Subsection~\ref{sec:continuous_entropy} and a distance in Subsection~\ref{subsec:Wasserstein}.
We then precisely define the gradient flow for the Kimura Equation in Subsection~\ref{ssec:Grad_Flow_form}.

\subsection{Background}\label{subsec:background}

The Kimura Equation is a particular example of a 
stochastic differential equation of the form
\[
\rd X_t =\sigma(X_t)\rd W_t +\mu(X_t)\rd t\ ,
\]
$W_t$ denoting the standard one-dimensional Brownian, cf.~\cite{EthierKurtz,champagnat2006unifying} and the study of this class is of paramount importance in population genetics, see e.g.~\cite{ewens2004mathematical}.

One particular feature of the models that we are interested in is that the dispersion coefficient $\sigma(x)$ is positive in the interior $x\in(0,1)$, but with $\sigma(0)=\sigma(1)=0$, representing two absorbing states on the boundary. 

For the resulting PDEs, the presence of those two absorbing states makes the notion of measure-valued solutions natural and actually necessary, and the probability laws typically take the form
\begin{equation}
\label{eqn:pform}
 p(t)=a(t)\delta_0+\widetilde p(t,x)\rd x+b(t)\delta_1.
\end{equation}
\begin{remark}
 In the same spirit as in the previous sections, the continuous part $\widetilde p(t,x)$ corresponds to the previous transient densities $N\widetilde{\bp}=(N p_i)_{i=1,\dots, N-1}$, while the boundary contributions $a\delta_0,b\delta_1$ correspond to the previous absorbed states $p_0,p_N$. 
\end{remark}
In this work, we call the Kimura Equation a generalised version of Equation~\eqref{eq:KimuraPDE}, i.e,
given by
\begin{equation}
\label{eq:Kimura}
\partial_tp=\frac{\kappa}{2}\partial_{xx}^2\left(\degdiff(x)p\right)+\partial_x\left(\degdiff(x) \partial_x V(x)\,p\right)\ ,
\qquad
t>0, \, x\in [0,1]
\end{equation}
where $\kappa>0$ is a diffusion coefficient and $V:[0,1]\to\R$ is the gradient potential as defined in Subsection~\ref{subsec:state of art}.
Furthermore, 
\begin{equation}
\label{eq:assumptions_theta}
\left\{
\begin{array}{l}
\degdiff(x)>0  \quad \text{for }x\in (0,1)\ ,\\
\degdiff(0)=0 \quad \text{with }\partial_x\degdiff(0)>0\ ,\\
\degdiff(1)=0 \quad \text{with }\partial_x\degdiff(1)<0\ .\\
\end{array}
\right .
\end{equation}
The typical case arising e.g. in \eqref{eq:KimuraPDE} is $\degdiff(x)=x(1-x)$, see Subsection~\ref{subsec:state of art}.
The fact that the zeroes are simple is crucial: the random motion is strong enough to counteract the deterministic advection driven by the velocity field $-\degdiff(x)\partial_x V$, and hence the trajectories are absorbed in finite time almost surely.
In terms of PDEs, the diffusion is locally uniform in any compact set  $K\subset (0,1)$, but degenerates at the boundaries.

As is common, we shall refer to the operator
\begin{equation}
\label{eq:def_L_forward}
\mathcal{L}_\kappa p\bydef\frac{\kappa}{2}\partial_{xx}^2\left(\degdiff  p\right)+\partial_x\left(\degdiff  \partial_x V\,p\right)
\end{equation}
in~\eqref{eq:Kimura} as the forward operator, while we speak of its formal adjoint
\begin{equation}
\label{eq:def_L_backward}
\mathcal{L}_\kappa^\dagger \zeta=\degdiff \Big [\frac\kappa 2 \partial^2_{xx}\zeta - \partial_x V \partial_x\zeta\Big]
\end{equation}
as the backward operator.
If not required from the context, the index $\kappa$ will be omitted from the operators $\mathcal{L}$ and $\mathcal{L}^\dagger$.

According to \eqref{eq:Kimura} it is clear that the interior density $\tilde p(t,x)$ in \eqref{eqn:pform} evolves according to
\[
\partial_t\tilde p=\mathcal L\tilde p\ ,
\ \mbox{for}\ t>0,\,\,x\in(0,1).
\]
However, because the diffusion is degenerate $\degdiff(0)=\degdiff(1)=0$, the evolution problem~\eqref{eq:Kimura} cannot be supplemented with standard boundary conditions as usual.
Additional conservation laws must be used instead to make sense of the Cauchy problem, and those are reminiscent from the discrete conservation laws~\eqref{eq:discrete_conservation_laws} in Section~\ref{sec:CTDMC}. 
In \cite{ChalubSouza09b,ChalubSouza14a} two of the authors studied the forward equation, and obtained a characterisation of the boundary measures in terms of those conservation laws.
More precisely, let $F$ be the solution of 
\begin{equation}
\label{eq:def_fixation_F_PDE}
\left\{
\begin{array}{ll}
\frac{\kappa}{2}\partial_{xx}^2 F - \partial_x V\,\partial_xF=0 & \mbox{for }x\in (0,1),\\
 F(0)=0,\\
 F(1)=1\ .
\end{array}
\right.
\end{equation}
\begin{remark}
\label{rmk:fixation_continuous}
This eigenfunction $F(x)$ is known as the \emph{fixation probability}, which encodes the probability of the population ending with a homogeneous population of type \A, when starting from an initial ratio $x$ (hence $F(0)=0$ and $F(1)=1$ for the two absorbing states) -- see also remark~\ref{rmk:discrete_fixation} at the discrete level.
\end{remark}
Imposing two additional conservation laws for the total mass and fixation
\begin{equation}\label{eq:conservation_laws}
 \frac{\rd\,}{\rd t}\int_{[0,1]}1\,\rd p_t(x)=0\quad\text{and}\quad\frac{\rd\,}{\rd t}\int_{[0,1]} F(x)\,\rd p_t(x)=0
\end{equation}
makes Equation~\eqref{eq:Kimura} well-posed \cite{Chalub_Souza:CMS_2009}, given an initial condition $p^\ini\in\mathcal{BM}^+([0,1])$ -- the space of positive Radon measures in $[0,1]$.
We should stress that both integrals are computed in the closed interval $x\in [0,1]$, since the measure $p(t)$ 
may -- and typically does -- charge the boundaries.
With those conservation laws newly enforced, the resulting Cauchy problem is studied in \cite{Chalub_Souza:CMS_2009} by means of (weighted, singular) Sturm-Liouville theory.
There, it was proved that the transient component ($\widetilde{p}$ in Equation~\eqref{eqn:pform}) is smooth up to the boundary, $\widetilde p\in C^\infty(0,T;C^{k+1}[0,1])$ if $\degdiff,V\in C^k([0,1])$.
More importantly, the two global conservation laws~\eqref{eq:conservation_laws} are equivalent to the two local flux conditions
\[
a'(t)=\degdiff'(0)\widetilde p(t,0)
\quad \mbox{and}\quad 
b'(t)=-\degdiff'(1)\widetilde p(t,1)
\qquad\forall t>0,
\]
driving the loss of mass from the continuous inner (transient) component $\tilde p(t,x)\rd x$ towards the boundary (absorbing) points $a\delta_0+b\delta_1$. These flux conditions can be heuristically obtained by inserting the representation for $p$ given in Equation~\eqref{eqn:pform} into the conservation laws given by Equation~\eqref{eq:conservation_laws}:
\[
0=\frac{\rd}{\rd t}\int_{[0,1]}p =
\frac{\rd}{\rd t}\left(
a(t)+\int_0^1\tilde p(t,x)\,\rd x + b(t)
\right)
=
a'(t) + \int_0^1 \partial_t\tilde{p} + b'(t)
\]
and
\[
0= \frac{\rd}{\rd t}\int_{[0,1]}pF=
\frac{\rd}{\rd t}\left(
a(t)F(0)+\int_0^1 \tilde p(t,x)F(x)\,\rd x + b(t)F(1)
\right)
=
\int_0^1 \partial_t\tilde{p}F + b'(t)
\]
where we used that the fixation $F(x)$ satisfies by definition $F(0)=0$ and $F(1)=1$.
Using the evolution law $\partial_t\tilde p=\mathcal L \tilde p$ for the smooth interior part, the stationary elliptic equation \eqref{eq:def_fixation_F_PDE} for the fixation $F$, and integrating by parts in the inner integrals then gives
\begin{align*}
0&=a'(t) +\degdiff'(1) \tilde{p}(1,t)- \degdiff'(0)\tilde{p}(0,t)+ b'(t)\\
0&=\degdiff'(1) \tilde{p}(1,t) + b'(t)\ .
\end{align*}
For more details, see \cite{Chalub_Souza:CMS_2009}.

\subsection{Entropy}
\label{sec:continuous_entropy}

Observe that we introduced the discrete entropy~\eqref{eq:def_entropy_BGS_discrete} in order to cope with the irreducibility of the relevant Markov chain.
Therefore, we expect that simple definitions of the entropy (such as, e.g., $\int p\log p$) should not work either in the continuous setting.
In fact, our construction in the discrete case crucially relied upon the (sub)dominant characteristic triple $(\mu,\widetilde \bw,\widetilde \bz)$ of $\bM$ (definition~\ref{def:characteristic_triple}).
Therefore, our first step will be to introduce the corresponding continuous objects:
\begin{definition}[Continuous characteristic triple]
\label{def:characteristic_triple_continuous}
The \emph{characteristic triple} is the triple $(\lambda,w,z)$ defined as the principal eigenvalue/eigenfunctions $-\mathcal Lz=\lambda z$ and $-\mathcal L^\dagger w =\lambda w$ with homogeneous Dirichlet boundary conditions, i-e
\begin{equation}
\label{eq:def_principal_eigenvalue_forward}
 \left\{
\begin{array}{l}
 -\frac{\kappa}{2}\partial^2_{xx}(\degdiff z)-\partial_x\left(\degdiff \,z\,\partial_x V\right)=\lambda z \quad \mbox{ for }x\in (0,1)\\
 \lim\limits_{x\to\{0,1\}}\degdiff(x)z(x)=0
 \end{array}
 \right.
\end{equation}
and
\begin{equation}
\label{eq:def_principal_eigenvalue_backward}
\left\{
\begin{array}{l}
 -\frac{\kappa}{2}\degdiff \partial^2_{xx}w+\degdiff \partial_x V \partial_xw=\lambda w \quad \mbox{ for }x\in (0,1)\\
 w(0)=w(1)=0
 \end{array}
 \right.
 .
\end{equation}
We always choose $w,z$ to be positive in $(0,1)$ and normalised as
\begin{equation*}
 \int _0^1 z(x)\rd x=1,
 \qquad 
 \int_0^1 w(x)z(x)\rd x=1.
\end{equation*}
\end{definition}
The fact that those principal eigenfunctions/values are well-defined follows from standard spectral theory, after observing that the above Sturm-Liouville problems are of \emph{limit-circle-non-oscillatory type}, see e.g. \cite{Zettl_2005}.
The eigenvalue $\lambda>0$ will quantify the exponential decay of $\widetilde p (t)$.
Just as in the discrete case, we define the new reference measure
\begin{equation}\label{eq:def_pi_continuous}
\pi(x):=w(x)z(x)
\quad 
\mbox{and}
\quad
\pi=\pi(x)\rd x\ .
\end{equation}
We slightly abuse the notations and identify the measure $\pi$ to its density $\pi(x)$ with respect to the Lebesgue measure.
Equation~\eqref{eq:def_pi_continuous} should be be compared to the discrete counterpart~\eqref{eq:def_pi_discrete}.
Our normalisation of $z$ and $w$ yields $\int_0^1 \rd \pi(x)=1$, and we view $\pi\in\mathcal P(0,1)$ as a reference probability measure.
The measure $\pi$ only charges $x\in (0,1)$, but should in fact be viewed as a probability measure in the whole underlying space $[0,1]$.
A useful information on this reference measure will be
\begin{lemma}\label{lem:ratio_continuous}
\label{lem:pi=wz}
With the same notations, there holds
\begin{equation}
\label{eq:pi=wz}
w(x)z(x)=\frac{1}{C}\frac{w^2(x)e^{-2V(x)/\kappa}}{\degdiff(x)},
\end{equation}
where  $C=\int_0^1\frac{w^2(x)\e^{-2V(x)/\kappa}}{\degdiff(x)}\rd x>0$ is a normalising constant
such that $\int\rd\pi(x)=1$.
 \end{lemma}
\begin{proof}
Defining $Z(x):=\frac{w(x)\e^{-2V(x)/\kappa}}{\degdiff(x)}$ and exploiting $-\mathcal L^\dagger w=\lambda w$, straightforward algebra leads to $-\mathcal L Z=\lambda Z$.
Since $ Z>0$ satisfies the boundary conditions $\lim\limits_{x\to\{0,1\}}\degdiff(x) Z(x)=0$, we see by uniqueness of the principal eigenfunction in Definition~\ref{def:characteristic_triple_continuous} that necessarily $Z(x)=Cz(x)$ for some constant $C>0$.
Hence $wz =\frac{1}{C}w Z$ takes the desired form, and the value of $C>0$ follows from our normalisation $\int \rd\pi =1$.\qed
\end{proof}

Standard results from measure theory allow to decompose any probability $p\in\mathcal P([0,1])$ as
\[
p=a\delta_0 +\widetilde p + b\delta_1
\qquad \mbox{with }\widetilde p=p\measurerestr{(0,1)}.
\]
Whenever $\widetilde p=\widetilde p(x)\rd x$ is absolutely continuous and the $w$-weighted mass
\[
\langle w,\widetilde p\rangle :=\int_0^1 w(x)\widetilde p(x)\,\rd x
\]
is non-zero, we define the renormalised transient probability measure
\begin{equation}
\label{eq:def_renormalized_q_continuous}
q(x)\rd x
\bydef
\frac{w(x)\widetilde p(x)}{\langle w,\widetilde p\rangle}\rd x.
\end{equation}
Mimicking our definition in the discrete case, we define the relative entropy by setting
\begin{equation}
\label{eq:entropy_continuous} 
\mathcal H(q|\pi):=
\left\{
\begin{array}{ll}
\displaystyle{\frac {\kappa}{2} \int_0^1\frac{\rd q}{\rd \pi}(x) \log \left( \frac{\rd q}{\rd \pi}(x)\right)\rd\pi(x)} & \mbox{if }q\ll \pi,\\
 +\infty & \mbox{otherwise}.
\end{array}
\right.
\end{equation}
Whenever $q\ll \pi$ we have of course $\frac{\rd q}{\rd\pi}(x)=\frac{q(x)}{\pi(x)}$.
For the dynamics, the renormalised $q(t,x)$ variable will be derived from the original $p(t,x)$ probability through~\eqref{eq:def_renormalized_q_continuous}, but we will in fact use $q$ as a primary variable/unknown as before in the discrete setting.
Formula~\eqref{eq:entropy_continuous} accordingly defines the entropy on the whole space $q\in\mathcal P(0,1)$, whether $q$ actually arose from some $p$ or not.
Moreover, the $\kappa/2$ factor appears in~\eqref{eq:entropy_continuous} due to the particular diffusion scaling in~\eqref{eq:Kimura}.
The reference measure $\pi=\pi_\kappa$ and the entropy functional $\H=\H_\kappa$ both depend on $\kappa$.
We shall in fact take the so-called \emph{deterministic limit} $\kappa\to 0$ later on, but we dispense at this stage from tracking the $\kappa$-dependence.

Just as in the discrete case, one can define more general entropies in terms of the original $p$ measure, which we called previously the \emph{substochastic} or \emph{reducible} entropies.
Similarly to~\eqref{eq:def_entropy_tsallis_discrete} and~\eqref{eq:def_entropy_tsallis_p}, if $\phi:\R^+\to\R $ is a convex, superlinear function, those read
\begin{equation}\label{eq:alternative_form_entropy_q}
\mathcal G_\phi(q|\pi)
:=
\int_{(0,1)}\phi\left(\frac{\rd q}{\rd \pi}(x)\right)\rd\pi(x)\ .
\end{equation}
The above expression should be understood to be $+\infty$ whenever $q\not\ll\pi$, and makes sense for general $q$.
However, when $q=\frac{w\widetilde{p}}{\langle w,\widetilde{p}\rangle}$ is obtained as the $Q$-process corresponding to some $p$ with $\langle w,\widetilde{p}\rangle\ne0$ (which propagates from $t=0$ to later times as in the discrete case, see also \eqref{eq:deff_p_q_cont} below), and recalling that $\pi(x)=w(x)z(x)$, we abuse the notations and also express this same entropy in terms of the original $p$ measure as
\begin{equation}
\label{eq:alternative_form_entropy_continuous}
\mathcal G_\phi(p|\pi)
:=
\int_{(0,1)}\phi\left(\frac{\widetilde p(x)}{z(x)\langle w,\widetilde p\rangle}\right) w(x)z(x)\,\rd x.  
 \end{equation}
The continuous BGS entropy~\eqref{eq:entropy_continuous} corresponds of course to the particular choice $\phi(\eta)=\frac\kappa 2\eta\log \eta$, and one can also consider Tsallis entropies $\phi(\eta)=\frac{\kappa}{2}\frac{\eta^m-\eta}{m-1}$. 

\subsection{The Wasserstein distance}
\label{subsec:Wasserstein}
In this section we introduce a suitable Wasserstein distance in the space of probability measures $\mathcal P([0,1])$ that will allow to write the Kimura Equation as a gradient flow.
Due to the presence of the variable coefficient $\degdiff(x)$ in~\eqref{eq:Kimura}, this quadratic Wasserstein distance  will not be based on the usual Euclidean distance, but will rely instead upon viewing the underlying $\Omega=(0,1)$ as a suitable Riemannian manifold.
More precisely, we consider the Riemannian metric with scalar product on the tangent plane $T_x\Omega$ at a point $x\in\Omega$ induced by $\frac{1}{\degdiff(x)}$, namely the norm of a tangent vector $\zeta\in T_x\Omega$ is defined as
\[
|\zeta|^2_{T_x\Omega}:=\frac{|\zeta|^2}{\degdiff(x)}.
\]
The induced generalised Shahshahani distance is
\begin{equation}\label{eq:distance}
 \mathsf{d}^2(x,y)
 \bydef
 \mathop{\mathrm{inf}}_{\substack{\xi\in C^1([0,1];\Omega) \\ \xi(0)=x,\,\xi(1)=y}}\int_0^1|\xi'(t)|^2_{T_{\small \xi(t)}\Omega}\,\rd t
 =
 \mathop{\mathrm{inf}}_{\substack{\xi\in C^1([0,1];\Omega) \\ \xi(0)=x,\,\xi(1)=y}}\int_0^1\frac{|\xi'(t)|^2}{\degdiff(\xi(t))}\rd t
\end{equation}
for $x,y\in\Omega$.
As can be expected, this distance is well-behaved:

\begin{lemma}\label{lem:expression_distance}
Assume that $\degdiff$ satisfies~\eqref{eq:assumptions_theta}.
Then the infimum in~\eqref{eq:distance} is achieved for a unique constant-speed geodesic, $\mathsf{d}$ can be uniformly extended to $\overline\Omega^2$, and
 \begin{equation}
\label{eq:dist_explicit}
  \mathsf{d}(x,y)=\left|\int_x^y\frac{\rd u}{\sqrt{\degdiff(u)}}\right|
  \qquad \forall\ x,y\in[0,1].
 \end{equation}
Moreover, $\mathsf{d}$ defines a distance in $\overline{\Omega}$ and the metric space $(\overline{\Omega},\mathsf{d})$ is Polish.
\end{lemma}

\begin{proof}
 For $x,y\in\Omega$ in the interior, the existence of a unique minimising curve is a standard exercise in the calculus of variations and we omit the details.
 As in the proof of Lemma~\ref{lem:1st_order_expansion_WN}, we start by writing the Euler-Lagrange Equation
 $\frac{\rd\ }{\rd t}\left(\frac{\xi'(t)^2}{\degdiff(\xi(t))}\right)=0$ and conclude that the intrinsic speed is constant, i.e., $|\xi'(t)|^2_{T_{\xi(t)}\Omega}=\mathsf{d}^2(x,y)$.
 In particular $\xi'(t)$ never vanishes and~\eqref{eq:dist_explicit} immediately follows from the change of variables $u=\xi(t)$ in~\eqref{eq:distance}.
 From the fact that $x=0,1$ are simple zeros of $\degdiff$ as assumed in~\eqref{eq:assumptions_theta}, the extension to $\overline{\Omega}$ follows; finally, completeness is an easy consequence of the explicit representation~\eqref{eq:dist_explicit}. \qed
\end{proof}

It is worth pointing out that this Shahshahani distance is locally equivalent to the Euclidean one in the interior (i.e. $c_K|x-y|\leq \mathsf d(x,y) \leq C_K|x-y|$ in any compact set $K\subset \subset (0,1)$), but behaves differently close to the boundary (e.g. $\mathsf d(0,x)\sim \int_0^x\frac{\rd u}{\sqrt{u}}\sim\sqrt{x}$ for small $x$).
This is reflected in the behaviour of the Kimura Equation~\eqref{eq:Kimura}, which is locally uniformly parabolic in the interior, but degenerate at the boundaries.
With the Polish space $(\overline\Omega,\mathsf d)$ at hand, one classically defines the corresponding Wasserstein distance on the space of probabilities $\mathcal P(\overline\Omega)$ as
\begin{equation}\label{eq:cont_Wasserstein}
 \W^2(\mu,\nu)=\min_{\gamma\in\Gamma(\mu,\nu)}\iint_{\Omega^2}\mathsf{d}^2(x,y)\rd\gamma(x,y),
 \qquad \mu,\nu\in\mathcal P(\overline\Omega).
\end{equation}
Here $\Gamma(\mu,\nu)$ denotes the set of admissible \emph{transport plans}, i.e., the set of probability measures $\gamma\in\mathcal{P}(\overline\Omega\times\overline\Omega)$ with first marginal $\gamma_1=\mu$ and second marginal $\gamma_2=\nu$.
The superposition principle
\begin{equation}
\label{eq:superposition_principle} 
\mathsf d(x,y)=\W(\delta_x,\delta_y)
\end{equation}
gives the natural correspondence between the underlying Polish space $(\overline\Omega,\mathsf d)$ and the overlying Wasserstein space $(\mathcal P(\Omega),\W)$, and we refer to \cite{Villani_Topics,villani2008optimal,santambrogio2015optimal} for an extended account on the optimal transport theory and bibliography.

As in the discrete case, we have the dynamical representation
\begin{proposition}[Benamou-Brenier formula \cite{benamou2000computational,lisini2009nonlinear}]
\label{prop:Benamou_Brenier}
For $q_0,q_1\in\mathcal P(\overline\Omega)$ there holds
\begin{equation}
\label{eq:Benamou_Brenier}
\W^2(q_0,q_1)=\inf\limits_{q,\psi}\quad \int_0^1\int_\Omega \frac{|\nabla \psi_t(x)|^2}{\degdiff(x)}\rd q_t(x)\,\rd t,
\end{equation}
where the infimum runs over narrowly continuous curves $[0,1]\ni t\mapsto q_t\in \mathcal P(\overline\Omega)$ with endpoints 
$q|_{t=0}=q_0,\,q|_{t=1}=q_1$ and satisfying the continuity equation
\[
\partial_t q_t+\dive (q_t\nabla\psi_t)=0
\]
with zero-flux boundary conditions.
\end{proposition}
\noindent
This is the exact equivalent of the dynamical definition of the discrete Wasserstein distance -- Definition~\ref{def:disc_Wasserstein} -- where the discrete 
continuity equation appear (see \cite{Maas_JFA} for discussions).
In the Lagrangian action~\eqref{eq:Benamou_Brenier} the velocity-field $\bv=\nabla\psi$ is measured not with respect to the standard Euclidean norm, but rather with respect to the intrinsic Shahshahani metrics $|\nabla\psi|^2_{T_x\Omega}=\frac{|\nabla\psi|^2}{\degdiff(x)}$.
We refer to \cite{lisini2009nonlinear} for a discussion on Wasserstein distances with variable coefficients, and to \cite{villani2008optimal} for optimal transport on abstract Riemannian manifolds.
Since we called the underlying metrics $\mathsf d$ the generalised Shahshahani distance, we shall sometimes speak of the corresponding Wasserstein distance as the Wasserstein-Shahshahani distance.

From the works of Otto \cite{otto2001geometry} it is known that the Wasserstein distance endows $\mathcal P(\Omega)$ with a (formal, infinite-dimensional) Riemannian structure, see also \cite{Villani_Topics,santambrogio2015optimal} for a comprehensive introduction.
In our setting with the $\degdiff(x)$ intrinsic tensor, the (formal) gradient of a functional $\mathcal F(q)$ with respect to this Riemannian structure reads
\begin{equation}
\label{eq:formula_gradient_Wasserstein}
 \mathrm{grad}_{\mathcal W}\mathcal{F}(q)=-\partial_x\left(\degdiff\,q\,\partial_x\left(\frac{\delta\mathcal{F}}{\delta q}\right)\right),
\end{equation}
see \cite{lisini2009nonlinear}. 
Here $\frac{\delta \mathcal{F}}{\delta q}$ denotes the first variation computed in the usual Euclidean sense, e.g. if $\mathcal F(q)=\int E(q(x))\rd x +\int V(x)q(x)\rd x$ then $\frac{\delta \mathcal{F}}{\delta q}(x)=E'(q(x))+V(x)$.
This is the exact counterpart of the discrete formula \eqref{eq:formula_gradient_discrete_Wasserstein}, with the subtle difference that the geometry on the underlying space $\mathcal X=[0,1]$ is now encoded by the mobility $\degdiff(x)$, while the ``discrete geometry'' on $\mathcal X_N=\{i/N,\,\, i=0\dots N\}$ was previously encoded directly by the kernel $\bK$ in \eqref{eq:formula_gradient_discrete_Wasserstein}.
%
\subsection{Gradient flow formulation}\label{ssec:Grad_Flow_form}
We want to identify now the Kimura Equation \eqref{eq:Kimura} as a gradient flow, based on the formula~\eqref{eq:formula_gradient_Wasserstein}.
To this end we first need to retrieve the evolution equation for the rescaled $Q$-process $q(t,x)$.

From \cite{ChalubSouza_2017} we know that the absolutely continuous part $\widetilde p(t,x)$ satisfies $\partial_t \widetilde p=\mathcal L\widetilde p$ in the classical sense and remains smooth up to the boundary.
Since $\mathcal L^\dagger w=-\lambda w$ with zero boundary values we see that
\begin{align*}
\frac{\rd\ }{\rd t}\int_0^1 w(x) \widetilde p(t,x)\rd x &= \int_0^1 w(x)\partial_t\widetilde p(t,x)\rd x
=\int_0^1 w(x)\mathcal L\widetilde p(t,x)\rd x\\
&=
\int_0^1 \widetilde p(t,x)\mathcal{L}^\dagger w(x)\rd x
=-\lambda\int_0^1 \widetilde p(t,x) w(x)\rd x,
\end{align*}
hence the weighted mass decays exponentially 
\begin{equation}
\label{eq:deff_p_q_cont}
 \langle w,\widetilde p(t)\rangle=\e^{-\lambda t}\langle w,\widetilde p^\ini\rangle
\end{equation}
as in the discrete counterpart~\eqref{eq:expo_decay_p}.
Discarding the case of completely absorbed initial data (leading to a trivial stationary evolution $p(t)\equiv p^\ini=a^\ini\delta_0+0+b^\ini\delta_1$), we can assume that the transient dynamics never gets absorbed, $\langle w,\widetilde p(t)\rangle\neq 0$ for all $t>0$, and thus renormalise
\begin{equation}\label{eq:qsemdx}
 q(t,x) \bydef
 \frac{w(x)\widetilde p(t,x)}{\langle w,\widetilde p\rangle}
 = \e^{\lambda t}\frac{w(x)\widetilde p(t,x)}{\langle w,\widetilde p^\ini\rangle}\ .
\end{equation}
Note that {$q(t,\cdot)$} is a probability (density) by construction.
As in the discrete case, $q(t,\cdot)$ can be obtained as the law of the natural $Q$-process, i.e. the original stochastic process $X_t$ conditioned to non-extinction in infinite time -- see for example \cite{Manturov1991,champagnat2016exponential}.
Since $\widetilde p(t,x)$ is uniformly bounded, cf.~\cite{ChalubSouza_2017}, and the eigenfunction $w(x)$ satisfies Dirichlet boundary conditions, we see that $q$ automatically satisfies
$q(t,0)=q(t,1)=0$ on the boundary, see Remark~\ref{rmk:dirichlet=neumann}.

\begin{remark}
	\label{rmk:dirichlet=neumann}
	In optimal transport and Wasserstein gradient flows, the evolution takes place by construction in the space of probabilities $\mathcal P([0,1])$, and one therefore usually enforces no-flux boundary conditions in the PDEs so as to comply with the conservation of mass---see, however, \cite{figalli2010new} for an application to gradient flows with Dirichlet boundary conditions.
	However in our framework, since $\widetilde p$ is bounded and $w$ vanishes on the boundaries, our new variable $q=\frac{w\widetilde p}{\langle w,\widetilde p\rangle}$ should also vanish and the evolution~\eqref{eq:grad_flow_Q} is implicitly understood here with Dirichlet boundary conditions $q|_{\partial\Omega}=0$.
	This is of course not a contradiction: since $w$ and $\degdiff$ vanish linearly we see from~\eqref{eq:pi=wz} that $\pi=wz$ does too, and the effective flux in~\eqref{eq:grad_flow_Q} is
	\[
	\left.\degdiff q\partial_x\log\frac{q}{\pi}\right|_{\partial\Omega}
	=
	\left.\degdiff q\left(\frac{\partial_x q}{q}-\frac{\partial_x\pi}{\pi}\right)\right|_{\partial\Omega}
	= \underbrace{\degdiff\big|_{\partial\Omega}}_{=0} \partial_x q\big|_{\partial\Omega}
	- q\big|_{\partial\Omega} \underbrace{\left.\frac{\degdiff}{\pi}\partial_x\pi\right|_{\partial_\Omega}}_{\sim C_{0,1}\neq 0}
	\sim C_{0,1}{q}\big|_{\partial\Omega}
	\]
	on the boundaries $x=0,1$.
	Thus the usual no-flux condition is here equivalent to our Dirichlet condition for $q$.
\end{remark}

Let us now identify the evolution law for our new variable $q$.
From \cite{ChalubSouza09b} we know that $\wp$ satisfies $\partial_t\wp=\mathcal L\wp$ in the classical sense, whence
\[
\partial_tq=\frac{\e^{\lambda t}}{\langle w,\widetilde p^\ini\rangle}\left\{\lambda w\widetilde p+w\partial_t\widetilde p\right\}
=
\frac{\e^{\lambda t}}{\langle w,\widetilde p^\ini\rangle}\left\{(-\mathcal L^\dagger w)\widetilde p+w\mathcal L\wp\right\}.
\]
Substituting the explicit expressions~\eqref{eq:def_L_forward},\eqref{eq:def_L_backward} of $\mathcal L,\mathcal L^\dagger$, respectively, we find after a straightforward calculation
\[
 \frac{\langle w,\widetilde p^\ini\rangle}{\e^{\lambda t}}\partial_tq=
 \frac{\kappa}{2}\partial_x\left[\degdiff w \widetilde p\partial_x\log\left(
 \frac{\degdiff\wp e^{2V/\kappa}}{w}
 \right)
 \right]
=
 \frac{\kappa}{2}\partial_x\left[\degdiff w \widetilde p \,\partial_x\log\left(
 \frac{w\wp }{
 \frac{w^2 e^{-2V/\kappa}}{\degdiff}
 }
 \right)
 \right].
\]
Using Equation~\eqref{eq:qsemdx} and Lemma~\ref{lem:ratio_continuous} to identify $\pi=\frac{1}{C}\frac{w^2e^{-2V/\kappa}}{\degdiff}$ inside the logarithm, we get
\begin{equation*}
 \partial_t q =\partial_x\left[\degdiff \,q\,\partial_x\left(\frac{\kappa}{2}\log\left(
 \frac{q}{\pi}\right)
 \right)
 \right].
\end{equation*}
Computing the first variation $\frac{\delta \H}{\delta q}=\frac{\kappa}{2}\left(\log\left(\frac{q}{\pi}\right)-1\right)$ of the entropy~\eqref{eq:entropy_continuous} and applying formula~\eqref{eq:formula_gradient_Wasserstein} for the Wasserstein gradient, we finally recognise the gradient flow structure
\begin{equation}
\label{eq:grad_flow_Q}
 \partial_{t} q=
 \partial_x\left(\degdiff q\partial_x\left(\frac\kappa 2 \log\left(\frac{q}{\pi} \right) \right) \right)
 \qquad \Leftrightarrow \qquad
 \partial_{t} q = -\grad_\W\H(q|\pi)
\end{equation}
for the Kimura Equation.

\begin{remark}\label{rmk:entropy_decreasing}
For any solution $q$ of Equation~\eqref{eq:grad_flow_Q},
it is not difficult to check that $\mathcal{G}_\phi(q|\pi)$ given by Equation~\eqref{eq:alternative_form_entropy_q} is nonincreasing in time as  in Lemma~\ref{lem:entropy_CD} and  Proposition~\ref{lem:decrasing}.
In fact we can even compute the dissipation
\begin{align*}
 \frac{\rd\ }{\rd t}\mathcal{G}_\phi(q|\pi)&=
 \int_0^1\phi'\left(\frac{q(x)}{\pi(x)}\right)\partial_tq\,\rd x
 =\int_0^1\phi'\left(\frac{q(x)}{\pi(x)}\right)\partial_x\left(\degdiff q\partial_x\left(\frac\kappa 2 \log\left(\frac{q}{\pi} \right) \right) \right)\rd x\\
 &=-\frac{\kappa}{2}\int_{(0,1)}\degdiff(x)\phi''\left(\frac{q(x)}{\pi(x)}\right)\left[\partial_x\left(\frac{q(x)}{\pi(x)}\right)\right]^2\rd\pi(x)\le 0
\end{align*}
 because $\phi$ is convex and because $q$ satisfies homogeneous Dirichlet boundary conditions at the end points --- see Remark~\ref{rmk:dirichlet=neumann}.
\end{remark}

\section{The replicator dynamics}
\label{sec:replicator}
 
The replicator dynamics was introduced in \cite{TaylorJonker_1978} and termed so in \cite{schuster1983replicator}.
The model consists in an infinite population of $n$ possible types, with frequency $x_i$ of type $i$.
This dynamics is based upon a simple postulate: the per-capita growth rate $\left(\frac{\dot{x}_i}{x_i}\right)$ of $x_i$ is given by the difference between the expected fitness of type $i$ and the population average fitness, i.e.
\begin{equation}
\label{eq:re:trad}
\dot{x}_i=x_i\left(\psi_i(\bx)-\bar{\psi}(\bx)\right), \quad i=1,\ldots,n,
\end{equation}
where $\psi_i$ is the fitness of type $i$ when the populations is at state $\bx$, and $\bar{\psi}(\bx)=\sum_{i=1}^nx_i\psi_i(\bx)$.
More recently, \cite{HofbauerSigmund} popularised the Replicator dynamics under two-player games, i.e., with $\psi_i(\bx)=(\bA\bx)_i$ for a given $n\times n$ matrix $\bA$ --- typically $\bA$ is associated to payoffs of a $n$-strategy, two-player game.

Equation~\eqref{eq:re:trad} is a cornerstone of evolutionary game theory, and it has been discussed and reviewed in various  works \cite{hofbauer2003evolutionary,weibull1997evolutionary,sigmund1999evolutionary,HofbauerSigmund}.
Incidentally, as we shall review below, it is also associated with the vanishing viscosity limit $\kappa\to 0$ of the Kimura Equation~\eqref{eq:Kimura}.
It is worth noticing that, while Kimura Equation arises in the infinite population limit as a delicate balance between selection effects and genetic drift, 
deviations from this balance lead to either a pure diffusive model or to a hyperbolic one --- the latter arises from selection dominating the genetic drift in the large population limit, and as discussed in \cite{ChalubSouza14a} it is equivalent to the Replicator Equation.
We refer to \cite{ChalubSouza09b,ChalubSouza14a} for a discussion about the different scalings and corresponding limits, see also \cite{ChalubSouza16} for a discussion on the different regimes both in finite and infinite populations.

In what follows, we will consider
 the  case of   $n=2$ types only, and in this case we may write $x_1=X$ and $x_2=1-X$ and write   the generalised one-dimensional \emph{Replicator Equation}
\begin{equation}\label{eq:Replicator}
\dot X=-\degdiff(X)\partial_X V(X)
\end{equation}
with same coefficient $\degdiff(x)$ and potential $V(x)$ as in Section~\ref{sec:CTCMC}.
Through natural embedding of point sets into empirical probability measures, any $L$-tuple of solutions $X_1(t)\dots X_L(t)$ to the ODE~\eqref{eq:Replicator} immediately gives a (probability) measure-valued solution
\[
p(t)=\frac 1L\sum _{l=1}^L \delta_{X_l(t)}
\]
to the corresponding hyperbolic PDE
\begin{equation}
\label{eq:Replicator_PDE}
 \partial_t p=\partial_x\left(\degdiff p\,\partial_xV\right)\,,
\end{equation}
which we also call the Replicator Equation with a slight abuse of notation.
The \emph{characteristics} ODE is the Lagrangian counterpart of the Eulerian hyperbolic PDE.
Note that \eqref{eq:Replicator_PDE} is obtained formally by taking the diffusion $\kappa=0$ in the Kimura Equation~\eqref{eq:Kimura}.

Since~\eqref{eq:Replicator_PDE} is written in terms of the $p$ variable, one might wonder as in Section~\ref{sec:CTDMC} what the conditioned $Q$-process might be, and what the resulting dynamics would be.
The main difference is that, since $\kappa=0$ here, no random fluctuation arises and the process is purely deterministic, $\rd X_t=-\degdiff(X_t)\partial_xV(X_t)\rd t$.
Therefore by Cauchy-Lipschitz uniqueness of trajectories for the ODE~\eqref{eq:Replicator}, any Lagrangian particle $X(t)$ initially on the boundaries remains absorbed, while a particle starting from the interior cannot reach the boundaries in finite time.
For the PDE this implies that the mass of $p^\ini=a^\ini\delta_0+\wp^\ini(x)\rd x+b^\ini\delta_1$ initially on the boundaries remains absorbed, while only the transient mass can evolve in the interior: In other words the distribution remains of the form $p(t)=a^\ini\delta_0+\wp(t,x)\rd x+b^\ini\delta_1$ for $t>0$, with $\|\wp(t)\|_{L^1}=\|\wp^I\|_{L^1}$.
As a consequence absorption never occurs, mass no longer leaks from the interior to the boundaries, and the previous transient rescaling from Section~\ref{sec:CTCMC} now simply reads
\[
q=\frac{\wp}{\|\wp\|_{L^1}}.
\]
Up to the constant-in time scaling factor $\frac{1}{\|\wp(t)\|_{L^1}}=\frac{1}{\|\wp^\ini\|_{L^1}}$ we have thus $q=\widetilde p$ for the replicator dynamics, and in fact one should think of the replicator Equation~\eqref{eq:Replicator_PDE} as acting on the $q$ variable rather than on $p$.

That being said, we have two equivalent gradient flow formulations for the replicator dynamics:
\begin{enumerate}
 \item 
 It is well known that the Replicator ODE~\eqref{eq:Replicator} is a gradient flow with respect to the (generalised) Shahshahani metric, see \cite{antonelli1977geometry,akin1979geometry,akin1990differential}.
 Indeed, choosing again to view $\Omega=(0,1)$ as a Riemannian manifold with the scalar product induced by $\frac{1}{\degdiff}$ on $T_x\Omega$ (see Subsection~\ref{subsec:Wasserstein}), an immediate computation allows to obtain the intrinsic gradients as $\grad_{\mathsf d}=\degdiff(x)\partial_x$, whence
 \[
 \dot X=-\degdiff(X)\partial_xV(X)
 \qquad \Leftrightarrow\qquad
 \dot X = -\grad_{\mathsf d} V(X)
 \]
 \item
 For the Replicator PDE~\eqref{eq:Replicator_PDE}, the Eulerian energy corresponding to the previous Lagrangian $V(X)$ for single particles is naturally
 \[
 \mathcal V(q)=\int_{0}^1V(x)\rd q(x),
 \qquad 
 q\in\mathcal P([0,1]).
 \]
 Our previous formula~\eqref{eq:formula_gradient_Wasserstein} for Wasserstein gradients gives next, with the first variation $\frac{\delta \mathcal V}{\delta q}=V$, the gradient flow structure
 \begin{equation}
 \label{eq:grad_flow_replicator}
 \partial_t q=\partial_x\left(\degdiff \,q \, \partial_xV\right)
 \qquad \Leftrightarrow\qquad
 \partial_t q= -\grad_\W\mathcal V(q).
 \end{equation}
\end{enumerate}

The convergence of the Kimura Equation~\eqref{eq:Kimura} towards the Replicator Equation~\eqref{eq:Replicator_PDE} in the deterministic limit $\kappa\to 0$ is well-known from a classical PDE point of view \cite{ChalubSouza09b,ChalubSouza14a}, but we will show in the next section that, using the right $q$ variable dictated by the conditioning of the corresponding $Q$-process, the convergence is variational (in some precise sense to be discussed later).
This is why we carefully and intentionally wrote the gradient flow~\eqref{eq:grad_flow_replicator} in terms of $q$ instead of using the original $p$ variable.
As just discussed, this is completely equivalent for the Replicator Equation (up to multiplicative scaling $p=q$), but the situation is drastically different in the presence of diffusion.
 %
 %
\section{Variational structures and their compatibility} 
\label{sec:gfcmp} 

So far, we have discussed three different gradient flow structures for models that are relevant to evolutionary biology: 
\begin{enumerate}[(i)] 
 \item 
the finite population dynamics discussed in Section~\ref{sec:CTDMC}, defined for population size $N<\infty$ 
 \item 
 the continuous population counterpart $x\in[0,1]$ discussed in Section~\ref{sec:CTCMC}, defined for $N=\infty$ and diffusion parameter $\kappa>0$ 
 \item 
 the replicator dynamics discussed in Section~\ref{sec:replicator}, which in its Eulerian formulation~\eqref{eq:Replicator_PDE} is defined for $N=\infty$ and $\kappa=0$ 
\end{enumerate} 
The convergence of (i) towards (ii) in the large population limit $N\to \infty$ as well as that of (ii) towards (iii) in the deterministic limit $\kappa\to 0$ have been proved to hold in particular models and in some appropriate sense (e.g. weak convergence of measures, or uniform convergence), see \cite{EthierKurtz,champagnat2006unifying,ChalubSouza09b,ChalubSouza14a}.
In this section we intend to convinde the reader that, under reasonable assumptions satisfied by many processes used in population genetics, our framework allows to further identify these convergences as natural \emph{variational} $\Gamma$-convergence of gradient flows, in a sense to be discussed shortly.
To the best of our knowledge this was never considered before for our classical triad: Not only do we provide a gradient flow structure for each of the three settings, but our structures are moreover energetically compatible with the relevant limits $N\to\infty,\kappa\to 0$. 
We should however stress that we do not aim at proving new convergence results here.
In addition, writing down a full, rigorous proof for $\Gamma$-convergence of gradient flow is usually a nontrivial task involving significant technical work.
Here we will only provide \emph{partial} results in this direction, and we will be content with the convergence of the driving energy (entropy) functionals and of the metric structures.

Perhaps, one of the surprising features of the discussion in the previous sections was its reliance on the $Q$-process variable $q$ for the analysis, since the latter is related to eternal paths in a system where the dynamics is almost surely absorbed in finite time. 
Traditionally, when dealing with models where absorption is certain, one relies on  quasi-stationary distributions (which happen to exist in most models of interest) in order to understand the fate of trajectories prior to absorption. 
However, when investigating a possible variational structure one should not expect this approach to be appropriate. 
Indeed, as already discussed, absorption is a non-reversible process, and reversibility was a key feature in obtaining a variational structure. 
Thus, for a generic trajectory that has not been absorbed at time $t$, the probability that it will remain non-absorbed decreases exponentially over time. 
When this probability decreases very slowly, meta-stable sates arise. 
However, even if these meta-stable states persist for very long times, the dynamics eventually becomes non-reversible in the long run. 
As a consequence one should not expect these trajectories to have a variational dynamics. 
These observations suggest that interesting trajectories should then be the immortal ones, i.e. those that never get absorbed.  
Two remarkable facts then happen: (i) this subset of trajectories is not empty, and one can indeed obtain a variational dynamics for these trajectories; (ii) the knowledge of the dynamics on this very restricted and small (zero-measure, negligible) set of trajectories is sufficient to recover the full transient dynamics (hence the whole dynamics, since the evolution of the absorbed states can be deduced from the transient dynamics with the help of the additional conservation laws).  
Roughly speaking, this is why one should rather consider the $Q$-process $q$ instead of the original $p$ distribution when seeking for a variational (gradient flow) structure, whether it be at the discrete or continuous level. 
 
\subsection{Gamma-convergence of gradient flows} 
\label{sec:Gamma} 
We first discuss shortly the notion of variational convergence of gradient flows needed for our purpose, and follow closely the exposition in \cite{sandier2004gamma,serfaty2011gamma}. 
Let us remind that the notion of $\Gamma$-convergence, introduced by E. De Giorgi in the 70's, is a notion of convergence of functionals that essentially guarantees convergence of the minimisers -- see the classical monograph \cite{dal2012introduction} for a detailed introduction. 
In some sense, this is precisely the notion of convergence needed when handling minimisation problems, and $\Gamma$-convergence is ubiquitous nowadays in variational analysis and modelling.
 
Often times one deals in practice with sequences of functionals that are not necessarily defined on the same space. 
In the following one should roughly keep in mind the idea of a $\Gamma$-converging sequence of functionals defined on a sequence of converging spaces. 
In order to illustrate this general idea, assume for simplicity that we are given a sequence of Hilbert spaces $\{H_\eps\}_{\eps>0}$ and a ``limit'' Hilbert space $H$, with some ``projections'' $\Pi_\eps:H_\eps\to H$. 
We say that a sequence $u_\eps\in H_\eps$ converges to $u\in H$ as $\eps\to 0$, denoted $u_\eps \overset{S}{\rightharpoonup} u$, if $\Pi_\eps(u_\eps) \xrightarrow[]{\sigma}u$ in $H$ for some topology $\sigma$. 
Both the projection $\Pi_\eps$ and topology $\sigma$ are crucial choices that one should make, depending on the model and applications under consideration. 
The Gamma-convergence of functionals on varying spaces is then defined as 
\begin{definition}[$\Gamma$-convergence] 
\label{defi:Gamma_CV} 
We say that a sequence of functionals $F_\eps:H_\eps\to\R\cup\{+\infty\}$ $\Gamma$-converges to $F:H\to\R\cup\{+\infty\}$ as $\eps\to 0$, denoted by $F=\mathop{\operatorname{\Gamma-\lim}} F_\eps$ or $F_\eps\overset{\Gamma}{\to} F$, if the following $\mathop{\operatorname{\Gamma-\liminf}}$ and $\mathop{\operatorname{\Gamma-\limsup}}$ conditions hold 
\begin{enumerate}[(i)] 
 \item 
 \label{item:defi_gamma_liminf} 
 for any sequence $u_\eps \overset{S}{\rightharpoonup} u$ there holds 
\begin{equation} 
\label{eq:def_Gamma_liminf} 
F(u)\leq \liminf\limits_{\eps\to 0} F_\eps(u_\eps). 
\end{equation} 
\item 
\label{item:defi_gamma_limsup} 
for any $u\in H$ there exists a \emph{recovery sequence} $u_\eps \overset{S}{\rightharpoonup} u$ such that 
\begin{equation} 
\label{eq:def_Gamma_limsup} 
\limsup\limits_{\eps\to 0} F_\eps(u_\eps) \leq F(u). 
\end{equation} 
\end{enumerate} 
\end{definition} 

Consider now the sequence of gradient flows $u_\eps(t):[0,T]\to H_\eps$ given by 
\begin{equation} 
\label{eq:eps_grad_flow} 
\partial_t u_\eps =-\grad_{H_\eps}F_\eps(u_\eps),  
\end{equation} 
where we emphasise the fact that the gradient of $F_\eps$ is computed with respect to the $H_\eps$ structure. 
Then, since gradient flows tend to minimise the energy along the evolution, and because $\Gamma$-convergence guarantees convergence of minimisers towards minimisers, one expects that limits of $F_\eps$-gradient flows should be gradient flows for the limiting functional $F=\mathop{\operatorname{\Gamma-\lim}}F_\eps$. 
(We shall refer to any such convergence as $\Gamma$-convergence of gradient flows.) 
This was proved in~\cite{sandier2004gamma,serfaty2011gamma} under additional conditions: 
\begin{theorem}[{\cite[Theorem 1]{serfaty2011gamma}}] 
\label{theo:gamma_CV_grad_flow_serfaty} 
Assume that $F=\mathop{\operatorname{\Gamma-\lim}}F_\eps$ for $\eps\to 0$ as in Definition~\ref{defi:Gamma_CV}, and let $u_\eps(t)$ be a solution of~\eqref{eq:eps_grad_flow} such that 
\begin{equation} 
\label{eq:convergence_curves_serfaty} 
u_\eps(t)\overset{S}{\rightharpoonup} u(t) 
\qquad \mbox{for all }t\in [0,T] 
\end{equation} 
for some limiting curve $u:[0,T]\to H$. 
If, 
additionally, $F(u(0))=\lim F_{\eps}(u_\eps(0))$, and if the lower bounds on the velocity and slope from \cite{serfaty2011gamma} hold, then $u$ is a solution of the limit gradient flow 
\[
\partial_tu=-\grad_H F(u). 
\] 
\end{theorem} 

We deliberately remain formal at this level and refrain from further discussing the precise definition of the above velocity and slope lower bounds, see again \cite{sandier2004gamma,serfaty2011gamma}. 
We should however stress that these two additional conditions are not a mere technical detail, and checking their validity is usually the most difficult part when trying to prove completely rigorous $\Gamma$-convergence of gradient flows.
Since our main concern here is the new biological paradigm rather than the careful and rigorous mathematical analysis, we deliberately choose to omit this technical part in order not to obfuscate the exposition.
However, we do not want to convey the wrong idea that $\Gamma$-convergence of the driving functional $F_\eps\xrightarrow[]{\Gamma}F$ and of the Hilbert spaces $H_\eps\to H$ are sufficient for the convergence of the associated gradient flows: There are of course counterexamples, but we believe that both the speed and slope conditions should hold in practice for our models in the large population and deterministic limits, $N\to\infty$ and $\kappa\to 0$.

Obtaining the convergence~\eqref{eq:convergence_curves_serfaty} of the sequence $u_\eps(t)$ towards \emph{some} limit curve $u(t)$ is in general not involved (standard weak compactness arguments typically apply), the challenge is rather to conclude that this limit is in fact a gradient flow for the limiting functional. (In such nonlinear settings this usually requires strong convergence.) 
 
One can actually build a theory of gradient flows in mere metric spaces (thus dispensing from any Hilbert or differential structures), as originally formulated by {De Giorgi} \cite{dg1980} in terms of \emph{curves of maximal slope}.
This is the notion we shall implicitly refer to in the sequel when we speak of \emph{variational evolution} or \emph{metric gradient flow}.
For the sake of exposition we refrain from discussing this delicate definition and refer instead to the classical monograph \cite{ambrosio2008gradient} (see also Subsection~\ref{sec:JKO} below). 
It was observed in \cite{sandier2004gamma} that the above scenario of $\Gamma$-convergence of gradient flows should hold in this very general metric setting, namely: if $(\mathcal X_\eps,d_\eps)$ is a sequence of metric spaces ``converging'' to a limit metric space $(\mathcal X,d)$ -- e.g. in the Gromov-Hausdorff sense -- then the $\Gamma$-convergence of the driving functionals $F_\eps:\mathcal X_\eps\to \R$ towards $F:\mathcal X\to \R$ should reasonably suffice (under additional speed and slope conditions) to guarantee the convergence of the corresponding metric gradient flows. 
Again, we should stress that this is not completely rigorous:
the additional speed and slope lower bounds from \cite{serfaty2011gamma}, actually required for a full convergence, correspond to suitable $\Gamma-\liminf$ estimates for the metric structure and metric derivative of the driving functional.
The former can sometimes be related to Gromov-Hausdorff convergence of metric spaces through Benamou-Brenier formulations, and both together provide suitable energy dissipation in the limit \cite{mielkeevo2016}.
Although the $\Gamma$-convergence of the functional and the Gromov-Hausdorff convergence of the underlying metric spaces alone do not suffice in general for the $\Gamma$-convergence of the gradient flows \cite{mielkeevo2016,MMP2020}, we claim that both the additional speed an slope convergences hold for our three specific evolutionary models, but we will not push further the rigorous mathematical analysis.

In the next section we shall exemplify this general scenario in two particular cases: 
\begin{enumerate} 
 \item  
In the limit of large populations $N\to\infty$ the space of discrete probability measures $\mathcal X_N=\mathcal P(\Delta_{N+1})$, endowed with the discrete Wasserstein distance $\W_N$ from Definition~\ref{def:disc_Wasserstein}, will Gromov-Hausdorff converge to the continuous Wasserstein space $\mathcal X=\mathcal P([0,1])$, endowed with the continuous Wasserstein-Shahshahani distance $\W$ from Section~\ref{sec:CTCMC}. 
 For the driving functionals we shall consider the sequence of discrete relative BGS entropies $H(\bq|\bpi_N)$ 
 that $\Gamma$-converge to the continuous counterpart $\mathcal H(q|\pi)$ 
 (up to scaling factors). 
 As a result and loosely speaking, the Moran process will $\Gamma$-converge to the Kimura model; more precisely, Equation~\eqref{eq:discrete_grad_flow_Q} will $\Gamma$-converge to Equation~\eqref{eq:grad_flow_Q}.
 \item 
 In the deterministic limit of small diffusion $\kappa\to 0$, the metric space will be fixed to be the Wasserstein space $\mathcal P([0,1])$, endowed with the fixed Wasserstein-Shahshahani distance $\W$, and we will consider the sequence of functionals 
 \begin{align*} 
  \H(q|\pi_\kappa) 
  &=\frac{\kappa}{2}\int_0^1 \frac{q(x)}{\pi_\kappa(x)}\log\left(\frac {q(x)}{\pi_\kappa(x)}\right)\rd\pi_\kappa(x) \\ 
  &=\frac \kappa 2 \int_0^1 q(x)\log q(x)\rd x +\int_0^1 V_\kappa(x)q(x)\rd x 
  \xrightarrow[\kappa\to 0]{\Gamma} 
  \mathcal V(q)\bydef\int_0^1 V(x)\rd q(x). 
 \end{align*} 
 Here $V(x)$ is the same potential initially prescribed for the Kimura Equation~\eqref{eq:Kimura}, and the effective potential $V_\kappa=-\frac{\kappa}{2}\log\pi_\kappa$ (to be defined below in more details) will converge to $V$. 
 As a result the Kimura gradient flow~\eqref{eq:grad_flow_Q} will $\Gamma$-converge to the Replicator gradient flow~\eqref{eq:grad_flow_replicator}. 
  
\end{enumerate} 
In both cases the $\Gamma$-convergence will be taken relatively to the weak-$*$ convergence of measures $q_n\weakstar q$, which we recall is defined by duality with bounded, continuous test-functions as $\int\varphi(x)\rd q_n(x)\to \int\varphi(x)\rd q(x)$ for all $\varphi\in\mathcal C_b$. 
This is a very reasonable choice, because the Wasserstein distance metrises the weak-$*$ convergence \cite{villani2008optimal}. 
 
For each case we will try to justify below why the above scenario should hold, namely we discuss the relevant $\Gamma$-convergences and the Gromov-Hausdorff convergence of the metric spaces. 
However, we will neither address the speed and slope lower bounds, nor discuss the well-preparedness of the initial data in Theorem~\ref{theo:gamma_CV_grad_flow_serfaty}. 
 We believe that those important assumptions hold true at least for the examples in Figure~\ref{fig:outline},
 but, as pointed out by Sandier and Serfaty in \cite{sandier2004gamma}, this is a case-to-case issue to be addressed by hand based on the specific structure of the problem under consideration. 
 We also refer to \cite{disser2015gradient} for the rigorous derivation of such lower bounds in the particular setting of finite-volume discretization of the Fokker-Planck Equation. 
For most of the classical processes, the convergence of solutions was proved most of the time in very strong topologies in previous works, using various techniques essentially based on PDE methods \cite{EthierKurtz,champagnat2006unifying,ChalubSouza09b,ChalubSouza14a}.
 Again, our interest does not lie here in rigorous proofs of convergence, and we rather wish to illustrate the energetic compatibility between the gradient flow structures. 
 
\subsection{The large population limit $N\to\infty$} 
With the particular Moran process and Kimura Equation in mind, we restrict here the general statements from Section~\ref{sec:CTDMC} to $k=2$ absorbing states. 
We chose to do so mainly for the ease of exposition, but the discussion below however extends to more general situations. 

In this particular setup the original $\bp$ variable is thus $(N+1)$ dimensional, and the rescaled $Q$-process $\bq$ is $(N-1)$-dimensional. 
The $(N+1)$-dimensional kernel $\bM$ gives rise to the rescaled $(N-1)$-dimensional kernel $\bK$, and we often write $\bL:=\bK^\dagger-\bI$ for the effective kernel driving the transient evolution $\frac{\rd\,}{\rd t} \bq =(\bK^\dagger-\bI)\bq =\bL\bq$ from Lemma~\ref{lem:three_equivalent}. 
 
Writing again $x_i=\frac{i}{N}$ for the uniform partition of $[0,1]$, we map canonically any $(N-1)$ probability vector $\bq$ to the associated empiric measure 
\[ 
\widehat q := \sum\limits_{i=1}^{N-1}q_i\delta_{x_i}\in \mathcal P([0,1]) 
\] 
(and in fact $\widehat q\in\mathcal P(0,1)$ since the absorbing states $i=0,N$ were discarded in the construction of the $Q$-process). 
Note that $\bq\mapsto\widehat q$ actually defines the ``projections $\Pi_\eps$'' from Subsection~\ref{sec:Gamma}, allowing to embed a sequence of varying discrete (probability) spaces into the limit (probability) space $\mathcal P([0,1])$. 
Slightly abusing the notations, we say that a sequence $\bq_N$ of probability vectors converges weakly-$*$ to the probability measure $q\in\mathcal P([0,1])$ as $N\to\infty$, denoted $\bq_N\weakstar q$, if $\widehat{q}_N\weakstar q$ weakly-$*$ in the sense of measures. 
Note that this defines the abstract convergence ``$\overset{S}{\rightharpoonup}$'' from Subsection~\ref{sec:Gamma}. 
 
We have then 
\begin{lemma} 
\label{lem:Gamma_CV_H}
Assume that $\bpi_N\weakstar \pi$, let $H(\bq|\bpi_N)$ be given by Equation~\eqref{eq:def_entropy_BGS_discrete} if $\bq\ll\bpi_N$ and $+\infty$ otherwise, and $\mathcal{H}(q|\pi)$ be given by Equation~\eqref{eq:entropy_continuous}.
Then 
\[
\H(\cdot\, |\pi) =\mathop{\operatorname{\Gamma-\lim}}_{N\to\infty}\,H(\cdot\,|\bpi_N)
\]
in the sense of Definition~\ref{defi:Gamma_CV}. 
\end{lemma} 
This only requires convergence of the reference measures $\bpi_N\weakstar \pi$, which should be satisfied in practice for reasonably physical models. 
A typical situation that we shall consider below is when the effective kernel $\bL_N=\bK_N^\dagger-\bI$ converges to the continuous operator $\mathcal L$, in which case one should also expect the eigenvectors $\bw_N,\bz_N$ to converge to the continuous eigenfunctions $w(x),z(x)$ and thus the measures $\bpi_N=\bw_N\circ\bz_N\rightharpoonup w(x)z(x)\rd x=\pi(x)\rd x=\pi$ at least in some weak sense. 
One can check that this holds at least for the convergence of the neutral Moran process towards the Kimura Equation, see Subsection~\ref{ssec:kimura_eq} below. 
\begin{proof} 
We need to check the two conditions in Definition~\ref{defi:Gamma_CV}, and we begin with the $\mathop{\operatorname{\Gamma-\liminf}}$ part. 
 Let $\bq_N\weakstar q$ be an arbitrary converging sequence, and notice that by definition we have $H(\bq_N|\bpi_N)=\H(\widehat q_N|\widehat\pi_N)$ with $\widehat q_N\weakstar q$ and $\widehat{\pi}_N\weakstar \pi$ in the sense of measures. 
 From the convexity and lower semi-continuity of $\eta\mapsto \eta\log\eta$ we can immediately apply  
 \cite[Theorem 2.34]{ambrosio2000functions} to conclude that the $\mathop{\operatorname{\Gamma-\liminf}}$ condition~\eqref{item:defi_gamma_liminf} holds as 
 \[ 
 \H(q|\pi)\leq  
 \liminf\limits_{N\to\infty}  \H(\widehat q_N|\widehat{\pi_N}) 
 =\liminf\limits_{N\to\infty} H(\bq_N|\bpi_N). 
 \]  
 For the $\mathop{\operatorname{\Gamma-\limsup}}$ part~\eqref{item:defi_gamma_limsup}, fix any $q\in\mathcal P([0,1])$. 
 If $\mathcal H(q|\pi)=+\infty$ there is nothing to prove, hence we can assume that $\mathcal H(q|\pi)<+\infty$ and in particular $q\ll\pi$ and $f:=\frac {\rd q}{\rd \pi}\in L\log L(\rd \pi)\subset L^1(\rd\pi)$. 
 By approximation it is enough to consider $f\in \mathcal C_b([0,1])$ positive, in particular $f\log f\in \mathcal C_b([0,1])$. 
 Let $f_N(x):=\frac{1}{\sum_{j=1}^N f(x_j)\pi_j}f(x)$ be the renormalisation of $f$, and observe that $f_N\to f$ uniformly because $\sum_j f(x_j)\pi_j=\int f(x)\rd\widehat \pi_N\to \int f(x)\rd\pi(x)=\int \rd q(x)=1$ (our main assumption is precisely that $\widehat \pi_N\weakstar \pi$). 
 Defining next the discrete probability vector $\bq_N$ by $q_i:= f_N(x_i)\pi_i$, we have of course $\bq_N\ll\bpi_N$, $\frac{\rd \bq_N}{\rd\bpi_N}=f_N$ and $\bq_N\weakstar f \pi=q$. 
 Moreover 
 \begin{align*} 
   H(\bq_N|\bpi_N)&=\sum\limits_{i=1}^{N-1} f_N(x_i)\log f_N(x_i)\pi_i=\int_0^1 f_N(x)\log f_N(x)\rd \widehat \pi_N(x) \\ 
   &\xrightarrow[N\to\infty]{}  \int_0^1 f(x)\log f(x)\rd \pi(x)=\mathcal H(q|\pi) 
 \end{align*} 
 as desired, where the convergence follows from our initial assumption that $\bpi_N\weakstar \pi $ (i-e $\widehat{\pi_N}\weakstar \pi$) and the strong uniform convergence $f_N\log f_N\to f\log f$ (the function $\eta\mapsto\eta\log \eta$ being uniformly continuous in any bounded interval $\eta\in[0,M]$). \qed
\end{proof} 
 
Turning now to the convergence of the metric spaces, N. Gigli and J. Maas proved in \cite{gigli2013gromov} that, if $\bL_N=\bK_N^\dagger-\bI$ is the standard Laplacian matrix in the discrete $d$-dimensional torus $\mathbb T_N^d=(\Z/N\Z)^d$ with mesh size $1/N$ (i-e if $\bK_N$ is the uniform random walk), then the discrete Wasserstein space $(\mathcal P(\mathbb T_N^d),\W_N)$ converges to the continuous Wasserstein space $(\mathcal P(\mathbb T^d),\W)$ in the sense of Gromov-Hausdorff. 
This is a non trivial result \emph{per se}, whose proof is the whole purpose of \cite{gigli2013gromov}. 
Let us mention that Gigli and Maas already pointed out potential applications to $\Gamma$-convergence of gradient flows, see also \cite{al201715,disser2015gradient} for a description of finite volume schemes for Fokker-Planck Equations as discrete Wasserstein gradient flows and their convergence towards continuous counterparts \cite{Jordan_Kindeleherer_Otto}. 
 
Here we do not pretend to prove any rigorous statement in this direction, and we shall be content with the following heuristics: 
\begin{claim} 
\label{claim:Gromov_Hausdorff} 
Let $\degdiff(x)$ satisfy our assumptions~\eqref{eq:assumptions_theta}, and $U\in\mathcal C^2(0,1)$. 
 If the effective kernel $\bL_N=\bK_N^\dagger-\bI$ arises from any reasonable finite difference discretization of the operator 
 \[ 
 \mathcal Lq:=\dive(\degdiff \nabla q) + \dive(\degdiff q \nabla U) 
 \] 
 on the domain $\Omega=(0,1)$ with no-flux boundary conditions and uniform mesh size $1/N$, then the discrete Wasserstein space $(\mathcal P_N,\W_N)$ Gromov-Hausdorff converges, as $N\to\infty$, to the continuous Wasserstein-Shahshahani space $(\mathcal P,\W)$. 
\end{claim} 
\begin{remark}
Similar questions were investigated in \cite{gladbach2018scaling,gladbach2019homogenisation}, where it was shown that some homogeneity and uniformity of the space meshing is essential (as assumed here). 
Note that part of our statement in Claim~\ref{claim:Gromov_Hausdorff} is that the limiting distance $\W$ does not ``see'' the potential $U$, while the discrete Wasserstein distance does depend on the whole kernel $\bK_N$ (hence a priori on $U$). 
This means that, as the population size increases, the discrete distance only retains the purely diffusive part $\dive(\degdiff \nabla q)$ of the elliptic operator, while the influence of the drift $\dive(\degdiff q\nabla U)$ smears out and vanishes in the limit. 
We believe that this should be the case in higher dimensions as well, see again \cite{gigli2013gromov} for a particular and rigorous $d$-dimensional statement in the torus. 
Let us point out that, in general, the discrete distance $\W_N$ is highly non local. 
Our assumption that the kernel arises from a finite difference approximation essentially localises the distance as $N\to\infty$: for example if the discretization is given by a fixed $m$-points stencil, then each discrete state interacts with a neighbouring fraction $m/N\to 0$ of all possible states (the kernel $\bK_N$ becomes sparser and sparser, and asymptotically concentrates on the diagonal). 
This dependency of the discrete metric on the potential can be seen as a drawback, and raises the question whether our processes possess more appropriate gradient structures, independently of the potential.
For the sake of exposition we prefer keeping this question open for future work, and refer to \cite{MMP2020} for recent results and related concepts of tilt and contact convergence of gradient flows.
\end{remark}
\begin{remark} 
Here we slightly abuse the notations: the operator $\mathcal L$ in our claim is not the forward operator~\eqref{eq:def_L_forward} acting on $p$ as before, but rather the divergence-form operator~\eqref{eq:grad_flow_Q} acting on $q$, namely $\mathcal L q=\partial_x\left(\degdiff q\partial_x\left( \log\left(\frac{q}{\pi} \right) \right) \right)=\partial_x(\degdiff \nabla q)-\partial_x(\degdiff q\partial_x\log \pi)$ with potential $U=-\log\pi$ and $\kappa=2$. 
\end{remark} 
 
 To give a hint of why the claim should hold, we remember Lemma~\ref{lem:1st_order_expansion_WN} and state that
\begin{corollary} 
\label{cor:infinitesimal_geometry_match} 
Let $\bL_N=\bK_N^\dagger-\bI$ be the three-points stencil, forward finite difference discretization of $\mathcal Lq=\partial_x(\degdiff \partial_x q) + \partial_x(\degdiff q \partial_x U)$ as in our Claim~\ref{claim:Gromov_Hausdorff}. 
For fixed $x\in(0,1)$ choose $i=i_N\sim\lfloor Nx\rfloor$ such that $x_{i}=\frac{i}{N}\to x$. 
Then 
\begin{equation} 
\label{eq:infinitesimal_variation_WN} 
\W_N(\be_{i},\be_{i+1})=\frac{1}{\sqrt{\degdiff(x)}}\times\frac{1}{N} +\mathcal O(1/N^2) 
\qquad 
\mbox{as }N\to\infty. 
\end{equation} 
\end{corollary} 
This means that the discrete Wasserstein distance fully encodes, at least asymptotically, the local geometry of the continuous Shahshahani space $(\Omega,\mathsf d)$ -- i.e. the Riemannian tensor $\frac{1}{\degdiff(x)}$ on $T_x\Omega$. 
Indeed formula~\eqref{eq:infinitesimal_variation_WN} gives the $\frac 1N$-infinitesimal variations of the discrete distance between point-measures located at $x_i\sim x$ and $x_{i+1}\sim x+1/N$. 
This is the exact counterpart of the continuous Wasserstein distance between point-masses 
\[ 
\W(\delta_x,\delta_{\small x+1/N}) 
=\mathsf d\left(x,x+1/N\right)=\frac{1}{\sqrt{\degdiff(x)}}\times \frac 1N +\mathcal O(1/N^2), 
\] 
which is easily checked using the representation formula~\eqref{eq:dist_explicit} for the Shahshahani distance. 
Since the overlying Wasserstein space $(\mathcal P([0,1],\W))$ is built upon -- and reflects the geometry of -- the underlying Shahshahani space $(\Omega,\mathsf d)$, this explains now our claim~\ref{claim:Gromov_Hausdorff}. 
 
\begin{proof}[Proof of Corollary~\ref{cor:infinitesimal_geometry_match}] 
In order to keep the notations light we write $x_i=i/N$, $\Delta x=x_{i+1}-x_i=\frac 1N$, $\degdiff_i=\degdiff(x_i)$, and $U_i=U(x_i)$. 
 Our assumption that $\bL$ is the forward finite difference approximation of $\mathcal Lq=\partial_x(\degdiff \partial_x q) + \partial_x(\degdiff q \partial_x U)$ means here that 
 \[ 
 (\bL\bq)_i= 
 \underbrace{\frac{ 
 \degdiff_i\frac{q_{i+1}-q_i}{\Delta x} 
 - 
 \degdiff_{i-1}\frac{q_{i}-q_{i-1}}{\Delta x} 
 }{\Delta x} 
 }_{\approx \,\partial_x(\degdiff \partial_x q)} 
 + 
 \underbrace{ 
 \frac{ 
 \degdiff_iq_i\frac{U_{i+1}-U_{i}}{\Delta x} 
 -\degdiff_{i-1}q_{i-1}\frac{U_{i}-U_{i-1}}{\Delta x} 
 } 
 {\Delta x} 
  }_{\approx \, \partial_x(\degdiff q\partial_x U)}. 
 \] 
 The off-diagonal coefficients of $\bK=\bL-\bI$ thus read 
 \[ 
 K_{\small i,i+1}=\frac{\degdiff_i}{\Delta x^2} 
 \] 
 and 
 \[ 
 K_{\small i+1,i}=\frac{\degdiff_i}{\Delta x^2} - \frac{\degdiff_i}{\Delta x}\cdot\frac{U_{i+1}-U_i}{\Delta x} 
 =\frac{\degdiff_i}{\Delta x^2}+\mathcal O(1/\Delta x), 
 \] 
 where the leading $\frac 1 {\Delta x ^2}$ order only stems from the higher-order diffusive part of the operator $\mathcal L$. 
 The lower $\mathcal O(1/\Delta x)$ order arises from the drift part exclusively, after absorbing one order of $1/\Delta x$ into $\left|\frac{U_{i+1}-U_i}{\Delta x}\right|\approx |\partial_x U|\lesssim C$ uniformly in $\Delta x$. 
 This explains why the limit distance $\W$ does not see the potential $U$ at leading order, and the same argument would carry through with any consistent discretization (centered, backwards, five-point stencil\dots) 
 Appealing next to Lemma~\ref{lem:1st_order_expansion_WN}, substituting the above expressions in~\eqref{eq:1st_order_expansion_WN}, and by scaling properties $\beta(\lambda x,\lambda y)=\lambda\beta(x,y)$ of the logarithmic mean, we finally get 
 \begin{align*} 
  \W_N(\delta_{x_i},\delta_{x_{\small i+1}})&=\int_0^1\frac{\rd r}{\sqrt{\beta\big(K_{\small i+1,i}r,K_{\small i,i+1}(1-r)\big)}} 
  \underset{N\to\infty}{\sim} 
  \int_0^1\frac{\rd r}{\sqrt{\beta\left(\frac{\degdiff _i}{\Delta x^2}r,\frac{\degdiff _i}{\Delta x^2}(1-r)\right)}}\\ 
  &=\frac{\Delta x}{\sqrt{\degdiff_i}}\underbrace{\int_0^1\frac{\rd r}{\sqrt{\beta(r,1-r)}}}_{=1} 
  \underset{N\to\infty}{\sim} 
  \frac{1}{\sqrt{\degdiff(x)}}\times\frac 1N\ , 
 \end{align*} 
and the proof is complete. \qed
\end{proof} 
It is worth noticing that all the computations in this proof are locally uniform in $i$ as long as $x_i=i/N$ remains bounded away from the boundaries, which in our statement was guaranteed since we considered $x_i\to x\in(0,1)$. 
Of course the behaviour close to the boundaries is drastically different because $\degdiff(x)$ vanishes and the diffusion degenerates.
 
\subsection{The deterministic limit $\kappa\to 0$} 
\label{sec:deterministic} 
Here we aim at recovering the Replicator Equation~\eqref{eq:Replicator_PDE} from the Kimura Equation~\eqref{eq:Kimura}. 
For both dynamics the metric space is fixed once and for all to be the Wasserstein-Shahshahani space $(\mathcal P([0,1]),\W)$, hence the only delicate point is to check the $\Gamma$-convergence for the sequence of functionals from Section~\ref{sec:CTCMC} when $\kappa\to0$. 
Emphasising now the dependence on $\kappa$, our relative entropy~\eqref{eq:entropy_continuous} was 
\[ 
\mathcal F_\kappa(q)\bydef
\frac{\kappa}{2}\H(q|\pi_\kappa) 
  =\frac{\kappa}{2}\int_0^1 \frac{q(x)}{\pi_\kappa(x)}\log\left(\frac {q(x)}{\pi_\kappa(x)}\right)\pi_\kappa(x)\rd x. 
\] 
Let us recall that we implicitly set $\mathcal F_\kappa(q)=+\infty$ whenever $q$ is not absolutely continuous with respect to $\pi_\kappa=\pi_\kappa(x) \rd x$ (or equivalently with respect to $\rd x$), and from Lemma~\ref{lem:pi=wz} the reference measure can be written 
\[ 
\pi_\kappa(x)=w_\kappa(x)z_\kappa(x)=\frac{1}{C_\kappa}\frac{w_\kappa ^2(x)e^{-2V(x)/\kappa}}{\degdiff(x)}. 
\] 
Here $C_\kappa$ is a normalising constant such that $\int \pi_\kappa =1$, and $w_\kappa,z_\kappa$ are the principal eigenfunctions of the operators $\mathcal L^\dagger_\kappa,\mathcal L_\kappa$ from Definition~\ref{def:characteristic_triple_continuous}. 
Using the above expression of $\pi_\kappa$ in terms of the fixed potential $V(x)$ in order to expand $\log\pi_\kappa$, we rewrite the entropy in the more convenient form 
\begin{align} 
\mathcal F_\kappa(q)&= 
\frac{\kappa}{2} \int_0^1 q\log q \,\rd x -\frac{\kappa}{2} \int_0^1 q\log \pi_{\kappa} \,\rd x 
\nonumber\\ 
&= \frac{\kappa}{2} \int_0^1 q\log q \,\rd x +\int _0^1 
\underbrace{\left( 
V+\frac{\kappa}{2}\log\left(\frac{C_\kappa\degdiff}{ w_\kappa^2}\right) 
\right)}_{\textstyle :=V_\kappa} 
q\, \rd x.
\label{eq:F_kappa_expanded} 
\end{align} 
 
As anticipated, we have next 
\begin{proposition} 
\label{prop:Gamma_liminf_deterministic} 
Assume that the sequence of probability measures $\{\pi_\kappa\}_{\kappa>0}$ satisfies a Large Deviation Principle (LDP) with speed $2\kappa^{-1}$ and rate function $V(x)$  in the deterministic limit $\kappa\to 0$, i.e. 
\begin{equation} 
\label{eq:LDP} 
-\inf\limits_{x\in \mathring{E}}\,V(x) 
\leq  
\liminf\limits_{\kappa\to 0}\frac{\kappa}{2}\log \pi_\kappa(E) 
\leq  
\limsup\limits_{\kappa\to 0}\frac{\kappa}{2}\log \pi_\kappa(E) 
\leq 
-\inf\limits_{x\in \overline E}\,V(x) 
\end{equation} 
for all Borel sets $E\subset [0,1]$, where $\mathring{E}$ and $\overline{E}$ are the interior and the closure of $E$, respectively.
For $q\in\mathcal P([0,1])$, define 
\[ 
\mathcal V(q):= 
\int_{0}^1 \{V(x)-\min V\}\, \rd q(x). 
\] 
Then 
\[
\mathop{\operatorname{\Gamma-\lim}}_{\kappa\to0}\mathcal F_\kappa=\mathcal V\ .
\]
\end{proposition} 
We refer to \cite{dembo1998large} for an introduction to Large Deviations. 
We stress that the problem strongly resembles the so-called \emph{entropic regularisation} in optimal transport, see e.g. \cite{carlier2017convergence,leonard2012schrodinger,cuturi2013sinkhorn}. 
 
Although not surprising at least formally in view of the expansion~\eqref{eq:F_kappa_expanded}, this convergence is not trivial in our particular setting: For any fixed $\kappa>0$ the eigenfunction $w_\kappa(x)$ vanishes at the boundary, hence the effective potential blows-up as $\lim_{x\to\{0,1\}} V_\kappa(x)=+\infty$. 
However, since the initial potential $V(x)$ is smooth up to the boundaries we see that the convergence $V_\kappa\to V$ cannot be uniform in $[0,1]$: the degenerate diffusion induces a thin boundary layer, whose length-scale should be accounted for by the $\Gamma$-convergence. 
Moreover, and even less evidently, there can exist internal transition layers near interior local maximum points of $V$ leading to metastable states. 
 
\begin{proof} 
With our strong assumption~\eqref{eq:LDP} this is an immediate consequence of the abstract results in \cite[\S 3]{mariani_GammaCV}. \qed
\end{proof} 
In practice, one should check by hand the strong hypothesis~\eqref{eq:LDP} in each case of interest, and this is a nontrivial task that strongly depends on the structure of the potential $V$. 
Let us point out that this trivially holds at least for the neutral Kimura Equation, i.e., Equation~\eqref{eq:KimuraPDE} with constant potential $V$.
In this case $\degdiff(x)=x(1-x)$ and $z_\kappa(x)=1,w_\kappa(x)=6 x(1-x)$, hence $\pi_{\kappa}(x)=w_\kappa(x)z_\kappa(x)=6 x(1-x)$ is independent of $\kappa$ and trivially satisfies the LDP~\eqref{eq:LDP}.
 
\subsection{Minimising movements and JKO schemes} 
\label{sec:JKO} 

In Section~\ref{sec:DTDMC}, we briefly discussed some popular models in finite population evolutionary dynamics, namely the Moran and the Wright-Fisher processes.
In these models, as time is a discrete variable, no gradient flow formulation is possible.
However, the so-called JKO scheme provides a direct view of the functional that is minimised by evolution, even when considering discrete in time models. In fact, one possible way to make sense of gradient flows in mere metric space is De Giorgi's \emph{minimising movement} \cite{de1993new}, which is roughly speaking an implicit Euler time-stepping. 
More precisely, given a driving functional $F$ in the abstract metric space $(\mathcal X,d)$, one wishes to make sense of $\partial_t x = -\grad _d F(x)$. 
To this end, choose a small time step $h>0$, and define the minimising scheme by solving recursively 
\[
x^{n+1}=\underset{x}{\operatorname{Argmin}} \left\{ 
\frac{1}{2h}d^2(x,x^n)+F(x) 
\right\}. 
\] 
This is sometimes called a discrete (in time) gradient flow: 
Indeed, for the particular instance of a Hilbert space $H$, the distance is $d^2(x,y)=\|x-y\|^2$, and the Euler-Lagrange Equation associated with $\min\limits_x\{\|x-x^n\|^2/2h +F(x)\}$ is nothing but the Euler implicit scheme $\frac{x^{n+1}-x^n}{h}=-\grad_H F(x^{n+1})$. 
Defining next the piecewise-constant interpolation in time 
$ x_h(t)\bydef x^{n+1}$ for $t\in (nh,(n+1)h]$ 
one should in general expect convergence $x_h(t)\to x(t)$ as the time step $h\to 0$, and obtain in the limit a solution to the abstract gradient flow. 
We refer to \cite{ambrosio2008gradient} for a detailed overview of this metric theory. 
 
When the space under consideration is that of probability measures endowed with the quadratic Wasserstein distance, Jordan, Kinderlehrer and Otto identified the classical Fokker-Planck Equation as the Wasserstein gradient flow of the BGS entropy \cite{Jordan_Kindeleherer_Otto}, precisely by proving the convergence of the minimising movement in the limit of small time steps. 
In this particular context this time discretization is often called the JKO scheme, after the three authors, and this has proved to be a powerful tool in order to prove existence of weak solutions of PDEs and construct numerical approximations.
 
Let us illustrate this general idea in the simplest possible case here, namely the one-dimensional replicator ODE~\eqref{eq:Replicator}. 
In this context, the minimising movement reads 
\begin{equation} 
\label{eq:JKO_Replicator} 
X^{n+1} 
\in\mathop{\mathrm{arg\,min}}_{X\in[0,1]}\left\{\frac{1}{2 h}\mathsf{d}(X,X^{n})^2+V(X)\right\}\ , 
\end{equation} 
where $\mathsf d$ is the Shahshahani distance from~\eqref{eq:distance} and we initialise $X^0=X^I$ for some given initial condition $X^I$. 
With our smoothness assumptions on the potential $V$, standard convexity arguments show that $X^{n+1}$ is uniquely well-defined, and as expected we have 
\begin{theorem} 
\label{theo:CV_JKO_replicator_ODE} 
Let $X_h(t):[0,\infty)\to [0,1]$ denote the piecewise-constant interpolation of the sequence $X^n$ defined in~\eqref{eq:JKO_Replicator}. 
Then $X_h(t)$ converges uniformly to $X(t)$ in any finite interval $t\in[0,T]$ as $h\to 0$, where $X(t)$ is the unique solution to the Replicator Equation~\eqref{eq:Replicator} with initial condition $X(0)=X^I$. 
\end{theorem} 
 
\begin{proof} 
Testing $X=X^n$ as a competitor in the variational scheme~\eqref{eq:JKO_Replicator} and summing over $n$, one gets the classical \emph{total square distance estimate} 
\[ 
\frac{1}{2h}\sum\limits_{n\geq 0}\mathsf d^2(X^{n+1},X^n)\leq V(X^0)-\inf\limits_{X\in[0,1]}V(X). 
\] 
This is a fairly general but crucial property of the abstract minimising movement, see \cite{ambrosio2008gradient}. 
One should think of this as an $H^1$ estimate $\int_0^\infty |\dot X_h(t)|^2_{X_h(t)}\rd t\leq C$, where the metric speed $\dot X_h(t)\in T_{X_h(t)}\Omega$ is measured in the intrinsic Shahshahani metrics. 
This in turn gives equicontinuity in time and thus compactness of $\{X_h\}_{h>0}$ in the uniform topology. 
Up to extraction of a subsequence we can therefore assume that $X_h$ converges uniformly to some $X$, and we only have left to prove that this limit curve is the unique solution to the Replicator ODE. 
 
To this end we use the definition~\eqref{eq:distance} of the Shahshahani distance to differentiate $\mathsf d^2(\cdot,X^n)$ and write the optimality condition for~\eqref{eq:JKO_Replicator} as 
\[ 
0=\left. 
\frac{d}{d X}\right|_{X^{n+1}}\left(\frac{1}{2 h }\mathsf{d}(X,X^n)^2+V(x)\right) 
= 
\frac{\mathsf{d}(X^{n+1},X^n)}{h\sqrt{\degdiff(X^{n+1})}}+V'(X^{n+1})\ . 
\] 
Exploiting formula~\eqref{eq:dist_explicit} and applying the mean value theorem, there exists $X_*\in[X(t),X(t+h)]$ such that 
\[
 \frac{X(t+h)- X(t)}{\sqrt{\degdiff(X_*)}}=	\int_{X(t)}^{X(t+h)}\frac{\rd z}{\sqrt{\degdiff(z)}}= 
	\mathsf{d}(X^{n+1},X^n) 
	=-h\sqrt{\degdiff(X(t+h))}V'(X(t+h))\ . 
\]
Taking $h\to0$, and from the uniform convergence $X_h(t)\to X(t)$, we find the Replicator Equation~\eqref{eq:Replicator}, together with the initial condition $X_h(0)=X^\ini\Rightarrow X(0)=X^\ini$. \qed 
\end{proof} 
 
 Similarly, all models discussed so far can be obtained by recursive minimisation of the following JKO functionals (denoting the previous steps $\bq_0,q_0$ as parameters): 
The finite state model given by Equation~\eqref{eq:discrete_grad_flow_Q} is associated to the minimisation of the functional 
 \begin{equation} 
  \label{eq:JKO_CTDMC} 
 J_N(\bq\,; \bq_0)\bydef\frac{1}{2h}\W_N^2(\bq,\bq_0)+H(\bq|\bpi_N)\ .
 \end{equation} 
The generalised Kimura Equation given by~\eqref{eq:grad_flow_Q} is obtained via minimisation of the functional
 \begin{equation} 
  \label{eq:JKO_KimuraQ} 
  \mathcal J_\kappa(q\,;q_0)\bydef\frac{1}{2h}\W^2(q,q_0)+\H(q|\pi_\kappa) 
 \end{equation} 
and, finally, the hyperbolic Replicator PDE~\eqref{eq:Replicator_PDE} corresponds to
\begin{equation} 
  \label{eq:JKO_replicator_PDE} 
 \mathcal J(q\,;q_0)\bydef\frac{1}{2h}\W^2(q,q_0)+\mathcal V(q|\pi)\ .
\end{equation} 
From our previous discussions on the various $\Gamma$-limits and Gromov-Hausdorff convergence, the one-step operators converge accordingly: on the one hand $J_N(\cdot\,;\bq_0^N)\xrightarrow[]{\Gamma} \mathcal J_\kappa (\cdot\,;q_0)$ in the large population limit $N\to\infty$ (presuming that the previous step $\bq^N_0\weakstar q_0$), and on the other hand $\mathcal J_\kappa(\cdot;q_0) \xrightarrow[]{\Gamma} \mathcal J(\cdot;q_0)$ in the deterministic limit $\kappa\to 0$. 
Since $\Gamma$-convergence guarantees convergence of the minimisers towards minimisers of the limit functional, this means roughly speaking that the discrete-in-time gradient flows (JKO schemes) also converge. 

We finish this subsection with some words on the relevance of JKO schemes in biological modelling. On one hand, there is an old argument if there is a functional which is minimised by biological evolution, cf. e.g.~\cite{Behera_1996,Ewens_1992}; on the other hand, all models discussed here were introduced without any explicit reference to such functionals. 
We just showed that not only do such functionals appear in all the models discussed so far, but also that they naturally split into a ``free energy'' -- which nature tries to minimise in time -- and an inertial term making such changes more difficult.
These ideas will be further explored in a subsequent work.
 
\subsection{From Moran processes to the Kimura Equation}\label{ssec:kimura_eq} 
 
 From the discussion in Subsection~\ref{subsec:state of art}, it is clear that the Kimura Equation~\eqref{eq:KimuraPDE} can be derived from the Moran process~\ref{def:Moran_TM}. In this section, we will reproduce this known result (see~\cite{ChalubSouza09b}) with a different approach. In fact, we will show that the precise form of the degenerate diffusion in Equation~\eqref{eq:KimuraPDE} is a consequence of the limiting behaviour of the characteristic triple, both in the discrete and continuous cases.
 More precisely
 
 \begin{lemma}\label{lem:Thetax}
  Consider the Moran process and assume the weak selection principle~\eqref{eq:wsp}. Let $(\mu,\widetilde{\bw},\widetilde{\bz})$ be its characteristic triple. Consider the generalised Kimura Equation~\eqref{eq:Kimura} and its continuous characteristic triple $(\lambda,w,z)$.
  Assume further that 
  $\frac{1}{N}\widetilde{\bw}\circ\widetilde{\bz}^{-1}$  converges pointwise uniformly to  $\frac{w}{z}$.
  Then $\degdiff(x)=x(1-x)$.
 \end{lemma}

 \begin{proof}
 From Equation~\eqref{eq:ratio_ev} we first have
 \begin{align*}
  \frac{w_i}{Nz_i}&=C\binom{N}{i}^{-1}\frac{\prod_{j=1}^i(1-\tsp_j)}{\prod_{j=1}^{i-1}\tsp_j}
  =C\binom{N}{i}^{-1}\frac{\prod_{j=1}^i\frac{N-j}{N}\left[1+\frac{2j}{\kappa N^2}V'(j/N)\right]}{\prod_{j=1}^{i-1}\frac{j}{N}\left[1-\frac{2(N-j)}{\kappa N^2}V'(j/N)\right]}\\ 
&=C\binom{N}{i}^{-1}\frac{(N-1)!N^{i-1}}{N^i(N-i-1)!(i-1)!}\left[1-\frac{2(N-i)}{\kappa N^2}V'(i/N)\right]\prod_{j=1}^i\frac{1+\frac{2j}{\kappa N^2}V'(j/N)}{1-\frac{2(N-j)}{\kappa N^2}V'(j/N)}\\ 
&=C\frac{i(N-i)}{N^2}\left[1+\mathcal{O}\left(N^{-1}\right)\right]
\prod_{j=1}^i\left[1+\frac{2}{\kappa N}V'(j/N)+\mathcal{O}\left(N^{-2}\right)\right]\ . 
 \end{align*}
 Following~\cite{Manturov1991} we get next
 \begin{align*} 
 \prod_{j=1}^i\left[1+\frac{2}{\kappa N}V'(i/N)+\mathcal{O}\left(N^{-2}\right)\right]&=\e^{\frac{2}{\kappa}\int_{1/N}^{i/N}V'(y)\rd y}\left[1+\mathcal{O}\left(N^{-1}\right)\right]\\
 &=\e^{\frac{2}{\kappa}(V(i/N)-V(1/N))}\left[1+\mathcal{O}\left(N^{-1}\right)\right] \ .
\end{align*} 
On the other hand, from Lemma~\ref{lem:pi=wz}, it is clear that $\frac{w(x)}{z(x)}=C\degdiff(x)\e^{2V(x)/\kappa}$, for a possibly different constant $C$, and imposing $x=i/N$ we finish the proof. \qed
 \end{proof}

\begin{remark}
The  convergence assumed in Lemma~\ref{lem:Thetax} is usually not straightforward to verify, but it is easily checked in the neutral case.
\end{remark}

As a direct consequence of Lemmas~\ref{lem:expression_distance} and~\ref{lem:Thetax}, we recover that the distance between two populations at deterministic states  is given by the so-called Shahshahani distance 
 \begin{equation}\label{eq:distance_continuous} 
  \mathcal W(\delta_x,\delta_y)
  =\mathsf{d}(x,y)=\left|\int_x^y\frac{\rd z}{\sqrt{z(1-z)}}\right|=2\left|\arcsin\sqrt{y}-\arcsin\sqrt{x}\right|. 
 \end{equation} 
 This distance turns out to be the same as the ``genetic distance'' introduced by Edwards and Cavalli-Sforza in the 1970's \cite{cavalli1999genetics,edwards1984likelihood}--- see also \cite{akin1990differential,antonelli1977geometry}.

\begin{remark}
	We hasten to point out that, while the Gromov-Hausdorff convergence of the metrics $\mathcal W_N\to \mathcal W$ is neither necessary nor sufficient for the $\Gamma$-convergence of gradient flows discussed above.
	However, the compatibility of such metrics suggests that there are generalisations of the metric displayed in Equation~\eqref{eq:distance_continuous} that are both compatible and relevant to discrete and continuous non-deterministic models.
\end{remark}

Finally, we point out that in most cases of biological interest one expects that the leading eigenvectors of the discrete process will converge pointwise
uniformly to the eigenfuctions of the continuous operator, and this is sufficient  to verify the  assumptions in  Lemma~\ref{lem:Gamma_CV_H}.   In this situation, we will also have that $\mathcal{G}_\phi(p|\pi)=G_\phi(\bp|\overline\bpi)+\mathcal{O}\left(N^{-1}\right)$, provided $\phi$ is sufficiently regular --- e.g. of locally bounded variation in $\R^+$, cf. \cite{Chui1971}.

\section{Conclusion}\label{sec:conclusion}

This work was born from a crossbreeding between two unrelated research programs: (i) to clarify differences and similarities in the triad of evolutionary models, and hence to understand them  in a unified way;  
(ii) to investigate the existence and relevance of local maximisation principles for evolutionary models and, by extension, in evolutionary biology.  

From this viewpoint, the main consequence of this work is to show that local maximisation principles --- and, consequently, variational structures --- may be formulated for the main models in evolutionary dynamics and are compatible among themselves. 

Perhaps not surprisingly, the correct approach to address this question was to understand the gradient flow formulation of all considered processes.
This program  has been already carry out for the Replicator Dynamics in \cite{akin1979geometry,Shahshahani_1979}.
For finite populations, the gradient flow formulation introduced by Maas in \cite{Maas_JFA} was the basic tool used to obtain a similar formulation for the continuous in time Moran process. The latter was then used as a starting point to obtain  this formulation for the continuous process (Kimura Equation). In addition, following earlier work \cite{ChalubSouza14a},  we embed the replicator dynamics in the hyperbolic equation formally obtained from the Kimura Equation in the vanishing genetic drift limit, and verify that the gradient flow structure found in \cite{Shahshahani_1979,akin1979geometry} is preserved as expected. 

Once these gradient structures are well established, a very natural question arises next:
Since there is a limiting relation among the processes involved, are the gradient structures compatible with one another?
A very natural framework to answer this question was for us the use of   $\Gamma$-convergence of gradient flows.  

The current work can also be seen as  part of a long program on the geometrisation of evolution initiated by \cite{akin1979geometry} (see also \cite{antonelli1977geometry,Shahshahani_1979,akin1990differential,Hofrichter_Jost_Tran}).
The gradient structure  for the continuous Moran process extends this  geometric approach to stochastic models for  finite populations without mutation.

As a byproduct, we introduced several possible free-energies (or entropies) for evolutionary processes.
A priori, there is no reason to favour any particular one, and it seems fruitful to  understand the dynamics induced by all these entropies in different scenarios (i.e., for several fitness potentials $V$), in particular the classical ones in evolutionary dynamics: dominance, convergence, and coexistence.
This will require a detailed study of the entropic dynamics of discrete and reducible Markov chains, not limited to the standard BGS entropy~\cite{Sober_2011}. 

An important question is how far the results obtained here can be extended. Ideally, there should be at least some relevant classes of more complex models (multi-type, structured) that should be amenable to a similar analysis.
An initial attempt towards generalisations appears in~\cite{ChalubSouza18}, with promising results; in particular the study of the fitness potential $V$ seems to provide hints on 3-types dynamics without mutations and in the 2-types dynamics with mutations.

Finally, it is important to point out that the optimisation process addressed here are all local --- i.e. they are \textit{myopic} in economics parlance --- and this does not guarantee the existence of  a global optimisation process. 
This framework is compatible with the adaptive dynamics point of view of evolution; see~\cite{Gyllenberg_Service_JM2011,Gyllenberg_MDL2011,meszena2002evolutionary,metz1996does,metz2008does,mylius2004does,gyllenberg2011necessary,kisdi1998frequency}; it is also compatible with Maximum Entropy Production Principles \cite{martyushev2006maximum,ziman1956general,jones1983variational}.
An extreme illustration of the difference between local and global maximisation in biological dynamics is the so-called \emph{evolutionary suicide}~\cite{Gyllenerg_Parvinen_2001}.

\begin{acknowledgements}
FACCC and AMR were partially supported by FCT/Portugal Strategic Project UID/MAT/00297/2019 (Centro de Matem\'atica e Aplica\c c\~oes, Faculdade de Ci\^encias e Tecnologia, Universidade Nova de Lisboa).
FACCC also benefited from an ``Investigador FCT'' grant. 
LM was supported by FCT/Portugal projects PTDC/MAT-STA/0975/2014 and PTDC/MAT-STA/28812/2017.
MOS was partially supported  by  CNPq  under grants \# 309079/2015-2, 310293/2018-9 and by CAPES — Finance Code 001.
\end{acknowledgements}

\begin{acknowledgements}
FACCC and AMR were partially supported by FCT/Portugal Strategic Project UID/MAT/00297/2019 (Centro de Matem\'atica e Aplica\c c\~oes, Faculdade de Ci\^encias e Tecnologia, Universidade Nova de Lisboa).
FACCC also benefited from an ``Investigador FCT'' grant. 
LM was supported by FCT/Portugal projects PTDC/MAT-STA/0975/2014 and PTDC/MAT-STA/28812/2007. MOS was partially supported  by  CNPq  under grants \# 309079/2015-2, 310293/2018-9 and by CAPES — Finance Code 001.
\end{acknowledgements}

\end{document}